\newif\iflipics
\newif\iflong 
\longtrue

\iflipics
\documentclass[a4paper,USenglish,cleveref,final,nolineno%
]{socg-lipics-v2021} 
\iflong

\hideLIPIcs
\makeatletter
\let\@oddfoot\@empty 
\makeatother
\fi

\else
\documentclass[11pt,a4paper]{article}
\usepackage[margin=24mm]{geometry}

\fi

\usepackage{amsmath,amssymb}
\usepackage{amsthm,amsfonts}

\usepackage{hyperref}
\hypersetup{hidelinks} 
\usepackage{graphicx}

\iflipics\else
\newtheorem{theorem}{Theorem}
\newtheorem{lemma}[theorem]{Lemma}
\newtheorem{proposition}[theorem]{Proposition}
\newtheorem{corollary}[theorem]{Corollary}
\fi

\newcommand\pforw{\vec{p}}
\newcommand\var{\mathrm{var}}
\newcommand\bal{\textit{bal}}

\newbox\backarrow
\setbox\backarrow=\hbox{\hskip 0,8pt 
  \rotatebox[origin=c]{180}{$\mathchar"017E$}}
\dp\backarrow=1pt 
\wd\backarrow=0pt
\ht\backarrow=2,2pt 
\newcommand\qback{\vbox{\copy\backarrow\nointerlineskip\hbox{$q$}}}

\title{Estimating the Number of Nodes and Paths in a DAG 
  by Following a Random Path}

\iflipics
\else
\author{Günter Rote\thanks{Institut für Informatik,
Freie Universität Berlin, Germany, 
\texttt{rote@inf.fu-berlin.de}, \url{https://orcid.org/0000-0002-0351-5945}}}
\date{October 31, 2026
}
\fi

\iflipics
\author{Günter Rote}{Institut für Informatik,
Freie Universität Berlin
}{rote@inf.fu-berlin.de}{https://orcid.org/0000-0002-0351-5945}{}

\authorrunning{G. Rote}

\Copyright{Günter Rote\iflong, \today\ (full version)\fi}



\iflong\else
\relatedversion{
  \url{%
https://github.com/guenterrote/estimate-DAG/blob/main/fullpaper.pdf%
}
\iflong (this version)\fi
}

\fi

\let\paragraph=\subparagraph

\fi

\begin{document}

\renewcommand\sectionautorefname{Section} 
\renewcommand\subsectionautorefname{Section} 

\maketitle
\begin{abstract}
  We investigate a simple method to compute an unbiased estimate of the
  number of nodes of a directed acyclic graph (DAG) by following a
  random directed path.  This method is due to Leonard Pitt from
  1987, and it generalizes 
  a
  method of D.~E.~Knuth from the 1960s for estimating the size of a
  rooted tree. A slight variation makes it possible to estimate the number of
  source-sink paths.
  \iflong\par\fi
  We extend the algorithm to use many ``agents'' that walk
  cooperatively through the graph. This extension unifies and generalizes several 
  methods that have been proposed for approximately counting the nodes of a tree.
  
\end{abstract}

\iflipics
\else


Enumerative combinatorics,
62D 
Statistical sampling 
\fi

\iflong
{

\subsection*{Table of Contents}

\def\addvspace#1{\vspace {1.5pt}} 

\def\section*#1{} 

\tableofcontents

}

\fi


\section{Problem statement and the Path Sampling algorithm (Algorithm~P)}
\label{sec:start}
We are given a
directed acyclic graph (DAG) $G=(V,E)$ and a \emph{cost function} $c\colon V\to \mathbb R$
on the
nodes.
The task is to 
estimate 
the total cost $c(G) := \sum_{v\in V} c(v)$.


The intended setting of the problem is that
 $G$ is too large to be fully scanned, but
we have access to the neighbors of a node and can thereby explore
parts of 
$G$.
In particular, for the basic Path Sampling algorithm, we make the following
assumptions:
\begin{itemize}
  \advance\itemsep -3pt
\item There is a single \emph{source} \iflong node \fi $s_0$ from which all other
  nodes can be reached by directed paths.
\item For any node, 
  we can pick a \emph{successor} node
\iflong (i.\,e., a neighbor via an outgoing arc) \fi
  uniformly at
  random.

\item For any node $v$, we can
  determine
  its outdegree $d^+(v)$,
  i.\,e., the number of its \emph{successors},
  its indegree $d^-(v)$,
  i.\,e., the number of its \emph{predecessors},
  and
  we can access 
  its cost $c(v)$.
\end{itemize}
The following algorithm, due to
Leonard Pitt \cite{Pitt1987},
computes an unbiased estimate $X$ of
$c(G)$ by walking along a directed path from the source.
\autoref{fig:example} illustrates the algorithm on a small example.

\enlargethispage{4pt}
 \medskip
\par
\begin{centering}
  \fboxsep=7pt
\noindent\fbox
{\begin{minipage}[c]{99em}
    \def\;{}
\begin{tabbing}
  \textbf{Algorithm P: Path Sampling}\\
   \qquad\=\+
   /\!/ Initialization: 
  \\ 
  $X := 0$\;\quad\= /\!/ accumulated estimate\\
  $u:=s_0$\;\> /\!/ current node, starts at the source\\
  $W := 1$\; \>/\!/ weight\\
  \textbf{loop}\iflong\\
  \qquad \fi\ \ \=\+
  /\!/ visit $u$:\\
      $X := X + W\cdot c(u)\;$ \ /\!/ augment the estimate
      \\
      \textbf{if} $u$ has no successor, terminate and \textbf{return} $X$ as the
      estimate\\
 /\!/ otherwise, proceed to a random successor $v$:\-\-\\
      (R)\>\>
 pick a node 
      $v$
      among the $d^+(u)$ successors of $u$  uniformly at random\\
      (U)\>\>\+\+
      $W:= W\times d^+(u)\,/\,d^-(v)\;$ \ /\!/ update the weight
      \\
      $u := v\;$ \ /\!/ go to $v$
\end{tabbing}
\end{minipage}%
}

\end{centering}


\begin{figure}[htb]
  \centering
  \hbox to \textwidth{\hfill\includegraphics[page=3,scale=0.97]{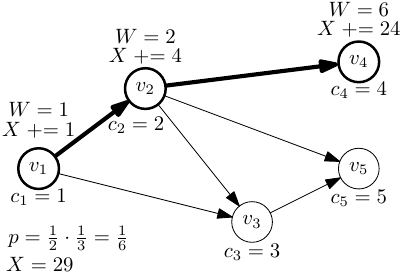}
    \hfill
    \includegraphics[page=2,scale=0.97]{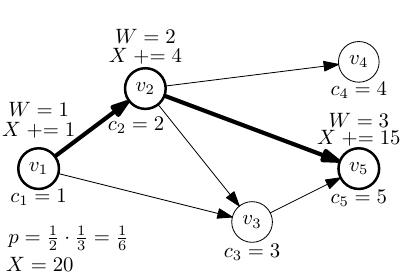}\hfill}
  \hbox to \textwidth{\hfill\includegraphics[page=4,scale=0.97]{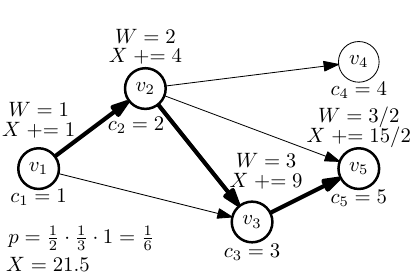}
    \hfill
  \includegraphics[page=5,scale=0.97]{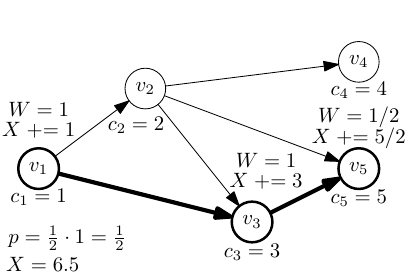}\hfill}
\caption{A 
  DAG $G$ with five nodes
    $v_1,\ldots,v_5$.
    The costs are $c_i=c(v_i)=i$, for a total cost of
    $c(G)=1+2+3+4+5=15$.
There are four sample paths that Algorithm P can take.    
    The expected estimate is $E[X]=
    \frac16\times 29 + 
    \frac16\times 20 + 
    \frac16\times 21.5 
    + 
    \frac12\times 6.5 
    = \frac{29+20+21.5+19.5}{6}
    =\frac{90}{6}
    =15
    $.
  }
  \label{fig:example}
\end{figure}


\paragraph{Path sampling for trees.}
Algorithm P was originally developed for rooted trees by D.~E.~Knuth in 1962.
It was first mentioned in a paper of Hall and Knuth
\cite{hall-knuth-1965} in 1965, and it has later been investigated
in its own right 
\cite{knuth75_estim},
  see also
%
\cite[Sect.~7.2.2, pp.~46--51,
Corollary~E]
{kn4b}.
For a tree, the indegree $d^-(v)$ is always 1, and the
weight $W$ is just the product of the outdegrees of the visited
nodes.

Algorithm P
is a straightforward extension to DAGs. 
To compensate for the fact that there can be several ways to reach
\iflong a given node \fi $v$,
we divide $W$ by the indegree $d^-(v)$ when
going to $v$.

\paragraph{Node costs.}
The original purpose and the most wide-spread application of Algorithm~P is
to estimate the cost of backtracking or branch-and-bound
algorithms.
The backtracking tree or branch-and-bound tree
that such an algorithm traverses may be 
enormous, while at other times, it is 
of moderate size and
the algorithm terminates 
reasonably fast.
If such a program has been running for a while, it would be useful to
know whether it is just about to finish or it would continue running
for years.
Therefore it is desirable to have a good prediction of the remaining
running time.
In these applications the cost $c(v)$ is the time to process node~$v$.

In other applications, the nodes represent some combinatorial
objects,
and one wishes to determine the number of these objects. In these situations,
one uses uniform costs $c(v)=1$ for all nodes~$v$.
For example, Avis and Devroye used the tree version of Algorithm~P to estimate
the number of vertices of a polytope~\cite{Avis-Devroye-2000}.

\iflong
If one wishes to count only the nodes with certain characteristics,
this can be arranged by choosing $c(v)\in\{0,1\}$.
 With appropriate
costs, 
 one can 
 gather statistics about various
properties of these objects,
 and of course, one can easily track several such 
 estimates simultaneously.
\fi

\subsection{Weighted selection}
\label{sec:weighted}
In line
(R),
the random successor does not have to be selected uniformly,
but we
can choose arbitrary positive transition probabilities $\pforw_{uv}$
subject
to the constraint that they form a probability distribution:
\begin{equation}
  \label{eq:forward}
  \sum_v \pforw_{uv}=1\text{, for all $u$ with $d^+(u)>0$}
\end{equation}
We call the quantities $\pforw_{uv}$
the \emph{forward proportions}.
\iflong
(The arrow above the $p$ 
does not signify a vector, but it reminds us of the direction
in which 
$\pforw_{uv}$ operates.)
\fi

Similarly, the incoming arcs at a node $v$ need not be weighted
equally, but they can be weighted by 
quantities
$\qback_{uv}$ that sum to 1 for the incoming arcs of each node~$v$:
\begin{equation}
  \label{eq:backward}
  \sum_u \qback_{uv}=1\text{, for all $v\ne s_0$}
\end{equation}
We call the quantities $\qback_{uv}$ the \emph{backward proportions}.
(They
\iflong
need not have a probabilistic interpretation, and they can even
\else
may
\fi
be negative.)

The update of the weight
when traversing the arc $uv$ in line (U) is replaced by
\begin{displaymath}
  \text {(U$_{q}$)}\qquad
  W:= W\times \qback_{uv}\,/\,\pforw_{uv} .
\end{displaymath}
\iflong The standard choices are
$\pforw_{uv}=1/d^+(u)$
and $\qback_{uv}=1/d^-(v)$. 
\fi


\begin{theorem}
  \label{th:unbiased}
  For arbitrary positive forward proportions
  $\pforw_{uv}$
satisfying~\eqref{eq:forward}
  and {backward proportions}~$\qback_{uv}$ satisfying~\eqref{eq:backward},  
 Algorithm~P produces an unbiased estimate of $c(G)$, i.\,e.,
  \begin{displaymath}
    E[X] = c(G).
  \end{displaymath}
\end{theorem}
We will prove this below as a consequence of a more general statement.

\subsection{No backward proportions: counting paths}
\label{sec:paths}

If we omit the backward proportions, we get 
an estimate of the 
\emph{number of paths starting from the source}.
The weight update step (U) is replaced by
\begin{displaymath}
  \text {(U$_{1}$)}\qquad
  W:= W\,/\,\pforw_{uv} .
\end{displaymath}

\begin{theorem}
  \label{th:unbiased-unweighted-paths}
  For arbitrary positive forward proportions
  $\pforw_{uv}$
satisfying~\eqref{eq:forward},
Algorithm~P
with update step \textup{(U$_1$)} produces an estimate $X$ whose expected
value is
  \begin{displaymath}
    E[X] = 
\sum_{v\in V} c(v) \cdot \#(\text{paths from $s_0$ to $v$ in $G$}).
  \end{displaymath}
\end{theorem}
\iflong As a special case, we can set
\else
Setting
\fi
$c(v)=1$ if $v$ is a sink
and  $c(v)=0$ otherwise\iflong , and we get
\else
\ gives
\fi
the number of source-sink paths.

\subsection{Arbitrary arc weights}
\label{sec:arbitrary-arc}

More generally, we can replace the backward proportions by arbitrary
real arc weights $e_{uv}\in \mathbb R$.
%
The weight function is extended from arcs to paths 
in the usual way: We
define $e(P)$ as the product
of the arc weights along the path~$P$:
\iflong
\begin{displaymath}
\else\begin{math}
  \fi
  e(P)=
  e(u_0u_1\ldots u_k) = e_{u_0u_1} e_{u_1u_2} \ldots e_{u_{k-1}u_k}
\iflong
\end{displaymath}
\else\end{math}.
  \fi


\begin{theorem}
  \label{th:unbiased-paths}
  For arbitrary positive forward proportions
  $\pforw_{uv}$
satisfying~\eqref{eq:forward}
and arbitrary arc weights $e_{uv}$,
Algorithm~P
with update step
\begin{displaymath}
  \textup {(U$_{e}$)}\qquad
    W:= W\times e_{uv}\,/\,\pforw_{uv} 
\end{displaymath}
produces an estimate $X$ whose expected
value is the quantity $c_e(G)$ defined as follows:
  \begin{displaymath}
    E[X] = c_e(G) :=
    \sum_{v\in V} 
    \
    \sum_{
      \textup{paths $P$ from $s_0$ to $v$}}
  e(P)c(v)
  \end{displaymath}
\end{theorem}

\iflong
\subsection{Analysis of the expectation}
\label{sec:analysis}
\begin{proof}[Proof of \autoref{th:unbiased-paths}]
\else
\begin{proof}
\fi
We define
  for every node $v$ the random variable
  \begin{displaymath}
    W_v :=
    \begin{cases}
      \text{the value of $W$ at the time that $v$ was visited,}\
      &\text{if  $v$ was
        visited}\\
      0, & \text{if $v$ was not
        visited}\\
    \end{cases}
  \end{displaymath}
  Then 
  $X = \sum_{v\in V} W_v\cdot c(v)$,
and therefore,
  $E[X] = \sum_{v\in V} E[W_v]\cdot c(v)$,
  by linearity of expectation.
  The theorem is a direct consequence of the following lemma.
  
  \begin{lemma} 
    For every node $v$,
  \begin{math}
E[W_v ] =
    \sum_{\textup{paths $P$ from
        $s_0$ to $v$}}
  e(P)
  \end{math}.
\end{lemma}
\begin{proof}
%
%
By definition, the expectation of  $W_v$ can be written as follows:
\begin{equation}
  \label{EW}
  E[W_v] = \sum_{\substack{\text{paths $P$ from $s_0$
        \ to $v$}}}
  \Pr[\text{path $P$ is taken}] \times (W_v\text{ as a result of path $P)$}
\end{equation}
For a path $P=u_0u_1\ldots u_k$, 
    \iflong   \begin{displaymath}  \else  \begin{math}\displaystyle
\fi      
  \Pr[\text{path $P$ is taken}]
  =
    \pforw_{u_0u_1} \pforw_{u_1u_2} \ldots \pforw_{u_{k-1}u_k}
\iflong.  \end{displaymath}      \else  \end{math}. \fi      
If the path $P$ is taken, then
\iflong   \begin{displaymath}  \else  \begin{math}\displaystyle \fi
  W_v =
  \frac{e_{u_0u_1}}{\pforw_{u_0u_1}}\cdot
  \frac{e_{u_1u_2}}{\pforw_{u_1u_2}}
  \cdots
  \frac{e_{u_{k-1}u_k}}{\pforw_{u_{k-1}u_k}}  
\iflong.  \end{displaymath}      \else  \end{math}. \fi      
Thus, in each product term in \eqref{EW},
the 
$\pforw\,$'s cancel, and the claim of the lemma follows.

The proof of
\autoref{th:unbiased-paths} is thus also complete.
\end{proof}
\setbox0=\hbox{\qedhere}
\end{proof}
\autoref{th:unbiased-unweighted-paths} is just a special case of
\autoref{th:unbiased-paths} with uniform arc weights $e\equiv 1$.

\begin{proof}[Proof of \autoref{th:unbiased} from \autoref{th:unbiased-paths}]
  Here the weights $e_{uw}$
  are the
  backward proportions $\qback_{uw}.$
  We only need to show the following statement:
 \setbox0=\hbox{\qedhere}
\end{proof} 
\begin{lemma} 
  Suppose the weights of the arcs $uw$ are chosen as $e_{uw}=\qback _{uw}$,
subject to~\eqref{eq:backward}.
 Then,
    for every node~$v$,
    \iflong
    \begin{displaymath}
      \else
    \begin{math}\displaystyle
\fi      
    \sum_{\textup{paths $P$ from
        $s_0$ to $v$}}
    e(P)
    =1
    \iflong.  \end{displaymath}      \else  \end{math}.\fi      
\end{lemma}
\begin{proof}
  This is easy to prove by induction on the length of the longest path
  from $s_0$ to $v$.

  For the source $v=s_0$, the claim follows
  directly
  from the fact that the degenerate path $P=s_0$ with no arcs has weight~1.

  For a node $v\ne s_0$,
  let
  $u_1,\ldots,u_d$ be the predecessor nodes of $v$,
  with $d = d^-(v)$.
We obtain
  \begin{align*}
    \sum_{\textup{paths $P$ from
        $s_0$ to $v$}}
    e(P)
&=    \sum_{i=1}^d \left(
    \sum_{\textup{paths $P'$ from
        $s_0$ to $u_i$}}
    e(P' )\qback_{u_iv}\right)
\\
&=    \sum_{i=1}^d \left(
    \sum_{\textup{paths $P'$ from
        $s_0$ to $u_i$}}
    e(P' )\right)\qback_{u_iv}
=    \sum_{i=1}^d
     \qback_{u_iv}
    =1,
  \end{align*}
  since
  each parenthesized expression in
  the last line is~1,
  by the induction hypothesis,
and the backward proportions
  sum to~1, by assumption~\eqref{eq:backward}.
\end{proof}
If the 
  backward proportions $\qback _{uw}$ are nonnegative, there is an easy alternate
  proof, following Pitt \cite{Pitt1987}:
  We interpret  $\qback _{uw}$ as the transition probability of
  going 
  from $w$ to $u$ in a random path in the reverse DAG.
  Then the expression $
    e(P)$ is the probability of traversing the path $P$
    in reverse
    when
    starting from node $v$.
The sum $\sum
    e(P)$ in the lemma is the probability of arriving at the source $s_0$, which is of course~1.

\section{Experiments}
\label{sec:examples}
To test the accuracy of the estimation, we tried the algorithm on two
families of
DAGs whose size we know.
Our goal was not a systematic computational study, but
to get a first impression of the behavior of the basic method and its variations.
\iflong
The reader who is more interested in the theoretical properties of the
estimate is encouraged to skip the experimental sections with the
tables.
\fi
The \textsc{Python} programs by which these preliminary experiments
were conducted
are available in
\iflong a \else the \fi public repository\iflong\footnote
{\else\ \fi
  \url{https://github.com/guenterrote/estimate-DAG
  }\iflong}\fi,
together with tables
of the results of experiments
over a wider range of parameters.

In our experiments, we have only estimated the number of nodes (not
paths), with the standard backward proportions $\qback_{uv}=1/d^-(v)$
and uniform costs $c(v)=1$.

Some more interesting potential applications of the method are mentioned in \autoref
{sec:applications}.

\subsection{The permutation DAG}
\label{sec:permutations}

The nodes of the \emph{permutation DAG} are
the $n!$ permutations of the 
elements $\{1,2,\ldots,n\}$.
The source node is the decreasing permutation $n,n-1,\ldots,3,2,1$,
and the successors of a permutation are the permutations that
can be reached by swapping two adjacent elements that are out of
order, as in bubblesort.
\iflong
The sink is the sorted permutation $123\ldots n$.
\fi
Figure~\ref{fig:perm-path} shows a
random 
path
\iflong
from the source 54321
\fi
to the sink 12345
for $n=5$, leading
to the estimate $X\approx 25.6296$. (The true size is $5!=120$.)

\begin{figure}[htb]
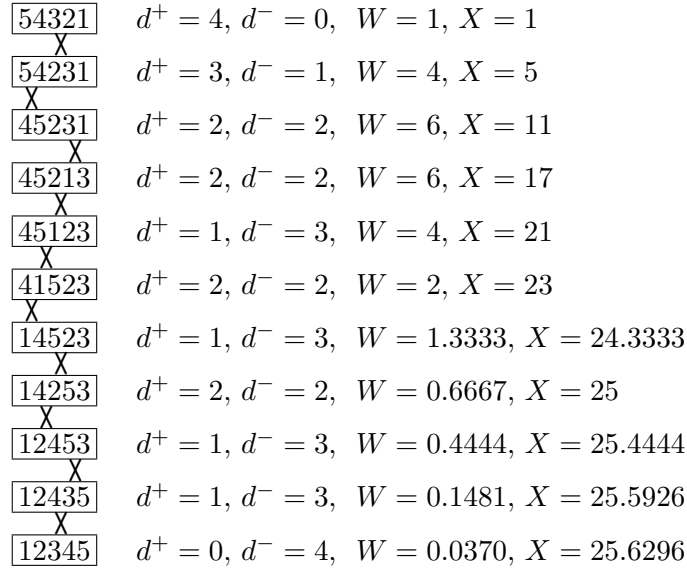

  \centering
  \def\cross{{\tt\rlap/\char`\\}}
  \fboxsep=2pt
  \def\place{\null\hskip\fboxsep\setbox0=\hbox{0}\hskip0,5\wd0 }
  \dimen0=-4pt
  \dimen1=-3.2pt
  {\begin{minipage}{90mm}
\fbox{54321} \quad $d^+=4$, $d^-=0$, \ $W=1$, $X=1$\\[\dimen0]
\place\phantom{00}\cross\\[\dimen1]
\fbox{54231} \quad $d^+=3$, $d^-=1$, \ $W=4$, $X=5$\\[\dimen0]
\place\phantom{}\cross\\[\dimen1]
\fbox{45231} \quad $d^+=2$, $d^-=2$, \ $W=6$, $X=11$\\[\dimen0]
\place\phantom{000}\cross\\[\dimen1]
\fbox{45213} \quad $d^+=2$, $d^-=2$, \ $W=6$, $X=17$\\[\dimen0]
\place\phantom{00}\cross\\[\dimen1]
\fbox{45123} \quad $d^+=1$, $d^-=3$, \ $W=4$, $X=21$\\[\dimen0]
\place\phantom{0}\cross\\[\dimen1]
\fbox{41523} \quad $d^+=2$, $d^-=2$, \ $W=2$, $X=23$\\[\dimen0]
\place\phantom{}\cross\\[\dimen1]
\fbox{14523} \quad $d^+=1$, $d^-=3$, \ $W=1.3333$, $X=24.3333$\\[\dimen0]
\place\phantom{00}\cross\\[\dimen1]
\fbox{14253} \quad $d^+=2$, $d^-=2$, \ $W=0.6667$, $X=25$\\[\dimen0]
\place\phantom{0}\cross\\[\dimen1]
\fbox{12453} \quad $d^+=1$, $d^-=3$, \ $W=0.4444$, $X=25.4444$\\[\dimen0]
\place\phantom{000}\cross\\[\dimen1]
\fbox{12435} \quad $d^+=1$, $d^-=3$, \ $W=0.1481$, $X=25.5926$\\[\dimen0]
\place\phantom{00}\cross\\[\dimen1]
\fbox{12345} \quad $d^+=0$, $d^-=4$, \ $W=0.0370$, $X=25.6296$
  \end{minipage}}
  \caption{A path in the permutation DAG}
  \label{fig:perm-path}
\end{figure}

We have 
tested the algorithm for the permutation DAG for $n=8$ and $n=9$.
\autoref{tab:P-PERM} displays the average of up to $m=10^7$ repetitions of
Algorithm~P.
The true value $n!$ is displayed in the last row.
\iflong We see that the
\else
The
\fi
values gradually converge to an estimate that is in the right ballpark.
\iflong\else\looseness -1 \fi

We also show the empirical standard deviation
\iflong
(the square root of the
empirical variance)
\fi
of the first $m$
samples, as well as the value of the largest sample.
The last row gives the precise theoretical
value of the standard
deviation
of the random variable~$X$, computed by the method of
\iflong
\autoref{sec:variance}.
\else
\autoref{sec:variance-whole}.
\fi

\iflong We see that the
\else
The
\fi
standard deviation is rather high, at least one order
of magnitude higher than the mean.
We also see that the largest estimate is very much larger than the
mean.
After $m=1000$ estimates for $n=8$, we have seen an estimate of around
$4.83
\times 10^7$, leading to the outlier mean of 77,136.
These phenomena indicate that the estimate is
unreliable.
Indeed,
in the context of backtracking trees,
it has often been observed that, while the estimate produced by Path
Sampling is unbiased, it has a high uncertainty, 
and the information it
gives is sometimes 
useless.
\iflong\else\looseness -1 \fi

For such situations,
 Chatterjee and Diaconis
 \cite[Section~2
   ]{chatterjee-diaconis-18}
proposed to look at the quantity   
$$\hat\chi = \frac{\max\{X_1,\ldots,X_m\}}
{X_1+\cdots+X_m},$$
see also \cite[Eq.~\thetag{34}, p.~50]{kn4b}.
A value of $\hat\chi$ close to 1 indicates that the mean is
essentially determined by a single  sample that 
dominates the
other samples by far.
Chatterjee and Diaconis
argue that a reasonable policy is to trust the experiments when
$\hat\chi$ is small.
Our tables show the
Chatterjee--Diaconis score
$\hat\chi$ in the last column.
\iflong\else\looseness -1 \fi

\begin{table}[htb]
  \centering
  \par\noindent\begin{tabular}{|c|r@{}lll|r@{}lll|}
\hline
\multicolumn{9}{|c|}{Permutations}\\
\hline
&\multicolumn4{c|}{$n=8$}&\multicolumn4{c|}{$n=9$}\\
$m$ & mean &\,$\pm$\,std.\,dev.&\multicolumn 1 c{max.}&\multicolumn 1 {c|}{$\hat\chi$}& mean &\,$\pm$\,std.\,dev.&\multicolumn 1 c{max.}&\multicolumn 1 {c|}{$\hat\chi$}\\
\hline\vbox to 11pt{}%
$10^1$&$39{,}331.8$&$\,\pm\,59{,}061.9$& $1.42{\times}10^{5}$ & $0.36044$&$48{,}168.7$&$\,\pm\,56{,}348.3$& $1.68{\times}10^{5}$ & $0.34803$  \\
$10^2$&$28{,}772.5$&$\,\pm\,64{,}855.0$& $3.42{\times}10^{5}$ & $0.11875$&$134{,}918.1$&$\,\pm\,4.87{\times}10^{5}$& $4.21{\times}10^{6}$ & $0.31195$  \\
$10^3$&$77{,}136.0$&$\,\pm\,1.55{\times}10^{6}$& $4.83{\times}10^{7}$ & $0.62640$&$188{,}542.8$&$\,\pm\,1.33{\times}10^{6}$& $2.93{\times}10^{7}$ & $0.15514$  \\
$10^4$&$44{,}629.3$&$\,\pm\,6.57{\times}10^{5}$& $4.83{\times}10^{7}$ & $0.10826$&$221{,}679.5$&$\,\pm\,2.23{\times}10^{6}$& $1.49{\times}10^{8}$ & $0.06720$  \\
$10^5$&$42{,}382.8$&$\,\pm\,1.13{\times}10^{6}$& $2.89{\times}10^{8}$ & $0.06821$&$372{,}754.8$&$\,\pm\,1.75{\times}10^{7}$& $4.94{\times}10^{9}$ & $0.13250$  \\
$10^6$&$39{,}837.5$&$\,\pm\,8.32{\times}10^{5}$& $3.63{\times}10^{8}$ & $0.00912$&$347{,}645.5$&$\,\pm\,1.23{\times}10^{7}$& $4.98{\times}10^{9}$ & $0.01433$  \\
$10^7$&$40{,}820.6$&$\,\pm\,1.92{\times}10^{6}$& $5.16{\times}10^{9}$ & $0.01265$&$351{,}350.7$&$\,\pm\,2.59{\times}10^{7}$& $4.64{\times}10^{10}$ & $0.01320$  \\
\hline
\vbox to 11pt{}&$40{,}320\phantom{.0}$&$\,\pm\,\rlap{$1.642{\times}10^{6}$\ \ \ \#visited = 29}$&&&$362{,}880\phantom{.0}$&$\,\pm\,\rlap{$1.614{\times}10^{8}$\ \ \ \#visited = 37}$&&\\
\hline
\end{tabular}

  \caption{Algorithm P for the permutation DAG.
    The number of visited nodes is always $1+\binom n2$, equaling the number
    of layers, because
  all source-sink paths have the same length.}
  \label{tab:P-PERM}
\end{table}

\subsection{Cubes and Klee--Minty cubes}
We take the graph \iflong (i.\,e., the edge skeleton) \fi of the unit cube in $n$ dimensions and orient its
edges by some generic objective function.
It turns out that,  on these instances, Algorithm~P works
impeccably! At distance $j$ from the source, all
nodes have indegree $j$ and outdegree $n-j$. Hence the accumulated
contributions to $X$ are $1+n+\binom n2 +\iflong \binom n3 + \fi \cdots + \binom
nj + \cdots = 2^n$.
We always get the correct answer.\iflong\else\looseness -1 \fi

\label{sec:klee-minty}

Since the behavior of Algorithm~P for standard 
cubes is so
uninteresting, 
we apply the algorithm
to \emph{Klee--Minty cubes} instead.
Klee and Minty, in 1972, defined a perturbation of the unit $n$-cube
for which they could show that the simplex algorithm, with certain
pivot rules, visits all $2^n$ vertices and thus takes exponential
time~\cite{klee-minty-1972}. 

\autoref{fig:klee-minty} shows a version
of this polytope with integral vertices, inscribed in a
box with exponentially increasing side lengths rather than a cube.
The arcs are oriented in the direction of the last coordinate $x_n$.
 This orientation of the arcs agrees with the original Klee--Minty
 cubes.
Our version of the Klee--Minty cube in $n$ dimensions is defined by the following linear program:
\iflong\else\looseness -1 \fi
\begin{displaymath}
  \text{Maximize } x_n
\end{displaymath}
subject to 
 the following
 $2n$ constraints:
 \begin{equation}
   \label{eq:KM}
\begin{aligned}
  0\le x_1 &\le 1\\
  x_{i-1} \le x_i & \le 2^i-1-x_{i-1}\text{, \ for $i=2,\ldots,n$}
\end{aligned}
 \end{equation}
We see that the $n$-dimensional Klee--Minty ``box'' has a \emph{floor} facet and a
\emph{ceiling} facet, where $x_n$ is at its lower bound $x_n=x_{n-1}$
or at its upper bound $x_n=2^n-1-x_{n-1}$, respectively.
Both the floor and the ceiling %
are affine images of an $(n-1)$-dimensional
Klee--Minty box. The floor has the original orientations, and in the
ceiling, all orientations are reversed.\footnote
{Our  model of the Klee--Minty cube is a slight variation of
the integral model proposed by
Kitahara and Mizuno~\cite {KITAHARA-Klee-Minty-2011},
where the floor facets lie in the coordinate hyperplanes.
The coordinates $x_i$ in~\eqref{eq:KM}
are related to
the coordinates $u_i$ in their model 
by the unimodular transformation
$x_i = u_1 + \dots + u_i$.}







\begin{figure}[htb]
  \centering
  \includegraphics{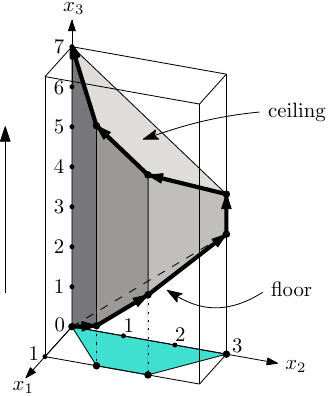}
  \caption{An integral Klee--Minty cube in $n=3$ dimensions with the
    directed path through all $2^n$ vertices. The edges are oriented upwards
    in the vertical $x_3$-direction.}
  \label{fig:klee-minty}
\end{figure}

\begin{table}[htb]
  \centering
  \par\noindent\begin{tabular}{|c|r@{}lll|r@{}lll|}
\hline
\multicolumn{9}{|c|}{Klee--Minty cubes}\\
\hline
&\multicolumn4{c|}{$n=10$}&\multicolumn4{c|}{$n=12$}\\
$m$ & mean &\,$\pm$\,std.\,dev.&\multicolumn 1 c{max.}&\multicolumn 1 {c|}{$\hat\chi$}& mean &\,$\pm$\,std.\,dev.&\multicolumn 1 c{max.}&\multicolumn 1 {c|}{$\hat\chi$}\\
\hline\vbox to 11pt{}%
$10^1$&$795.5$&$\,\pm\,2{,}170.0$& $6.96{\times}10^{3}$ & $0.87488$&$198.8$&$\,\pm\,335.8$& $1.05{\times}10^{3}$ & $0.53028$  \\
$10^2$&$277.6$&$\,\pm\,937.2$& $6.96{\times}10^{3}$ & $0.25073$&$771.3$&$\,\pm\,3{,}045.2$& $2.41{\times}10^{4}$ & $0.31309$  \\
$10^3$&$906.3$&$\,\pm\,11{,}238.4$& $2.85{\times}10^{5}$ & $0.31473$&$1{,}068.5$&$\,\pm\,10{,}497.3$& $2.81{\times}10^{5}$ & $0.26334$  \\
$10^4$&$485.6$&$\,\pm\,4{,}998.2$& $2.85{\times}10^{5}$ & $0.05874$&$1{,}902.3$&$\,\pm\,85{,}597.0$& $8.39{\times}10^{6}$ & $0.44087$  \\
$10^5$&$955.3$&$\,\pm\,90{,}694.9$& $2.81{\times}10^{7}$ & $0.29453$&$1{,}786.1$&$\,\pm\,60{,}458.6$& $9.70{\times}10^{6}$ & $0.05430$  \\
$10^6$&$795.1$&$\,\pm\,59{,}515.4$& $4.42{\times}10^{7}$ & $0.05565$&$2{,}206.7$&$\,\pm\,2.13{\times}10^{5}$& $1.52{\times}10^{8}$ & $0.06900$  \\
$10^7$&$909.0$&$\,\pm\,1.78{\times}10^{5}$& $4.89{\times}10^{8}$ & $0.05384$&$3{,}103.9$&$\,\pm\,2.13{\times}10^{6}$& $6.58{\times}10^{9}$ & $0.21208$  \\
\hline
\vbox to 11pt{}&$1{,}024\phantom{.0}$&$\,\pm\,\rlap{$3.246{\times}10^{9}$}$&&&$4{,}096\phantom{.0}$&$\,\pm\,\rlap{$1.034{\times}10^{19}$}$&&\\
&\multicolumn4{l|}{\#visited: avg\ = 15.7, max\ = 78}&\multicolumn4{l|}{\#visited: avg\ = 20.6, max\ = 100}\\
\hline
\end{tabular}

  \caption{Algorithm P for Klee--Minty cubes}
  \label{tab:P-KM}
\end{table}

Algorithm P essentially performs the simplex algorithm with the \emph{random-edge} pivot rule.
For the expected number of steps of this pivot rule on  Klee--Minty
cubes,
there is an easy upper bound of $O(n^2)$%
~\cite{Randomized-simplex-Klee-Minty-1998}.
A~matching lower bound
of $\Omega(n^2)$
is known only for a random start vertex~%
 \cite{Klee-Minty-random-edge-2007}.

\autoref{tab:P-KM} shows the results of two runs, with $n=10$ and
$n=12$ dimensions. In comparison to the permutation DAG, we see that
the estimate is less accurate, even though the cubes have far fewer nodes. The theoretical standard deviation is very large
(around $10^{19}$ for $n=12$), but this high standard deviation is not
observed empirically. \iflong
{The extreme clash between the predicted and the observed standard deviation
raises
the suspicion of an error in the formulas or a programming error.
However, for small values of $n$, the empirical standard deviation converges
very well to the predicted standard deviation.
The phenomenon is inherent in the data. The high value of the standard
deviation,
several orders of magnitude larger than the mean,
indicates that there is an unusual situation.
}
\fi

\iflong
\autoref{sec:KM-details} gives some details about
\iflong
the
representation that we used
\else
our representation
\fi
for the vertices of the
Klee--Minty cubes.
\fi

\iflong
\section{Computing the variance of \texorpdfstring{$X$}{X}}
\label{sec:variance}
\label{sec:variance-whole}

For test DAGs that are small enough so that one can iterate through
all nodes, it is possible to compute the variance of the random
estimate~$X$ for Algorithm~P exactly.

We use the following basic result about
a mixture of distributions:

\begin{lemma}[Mixture of distributions]
  \label{lem:mixture}
  Consider $d$ random variables $Z_1,\ldots,Z_d$ with means $\mu_i$ and
  variances $\var_i$.
  Let $Z$ be the mixture of these variables with proportions
  $\alpha_1,\ldots,\alpha_d$, that is, $Z$ is set to $Z_i$ with probability
  $\alpha_i$,
  where $  \sum_{i=1}^d \alpha_i=1$. Then $Z$ has mean
  \begin{displaymath}
    \mu = \sum_{i=1}^d \alpha_i\mu_i
  \end{displaymath}
  {and variance}
  \begin{equation}
\label{eq:mixture-variance}
    \var = \sum_{i=1}^d \alpha_i\bigl(\var_i+(\mu_i-\mu)^2\bigr)
        =         \sum_{i=1}^d \alpha_i\var_i
+ \sum_{1\le i<j\le d} \alpha_i \alpha_j(\mu_i-\mu_j)^2   .
  \end{equation}
\end{lemma}
These equations can be seen as
instances of
the \emph{law of
total expectation}
$E[Z]=E[E[Z|Y]]$
and
the \emph{law of
total variance}
$\var(Z)=
E[\var(Z|Y)] + \var(E[Z|Y]) 
$,
where $Y\in \{1,\ldots,d\}$ is a random variable that determines which
$Z_i$ is chosen.
The last expression in 
\eqref{eq:mixture-variance} is an alternative formula,
which
\iflong  is sometimes \else may be \fi more convenient 
for hand calculation when $d$ is small.

Let us apply this to  Algorithm~P.
For each node $u$, we define the random variable $D_u(W)$ to be the
 value of all future additions to~$X$, assuming
that $u$ is entered with weight $W$. Since the operations on $W$
consist of multiplications by random quantities that are independent
of~$W$, it is clear that
$D_u(W)$ is a linear function of
~$W$, i.\,e.,
\begin{displaymath}
  D_u(W) = F_uW,
\end{displaymath}
for some random factor $F_u$ that depends only on the graph~$G$,
the costs~$c(v)$, the forward proportions $\pforw_{uv}$, and the arc
weights.

We now show how the mean and the variance of the variables~$F_u$ can be computed
in 
reverse direction of the arcs,
proceeding from the end of the DAG towards the source.
Eventually we 
arrive at the final estimate
$X=F_{s_0}$, and $\var(F_{s_0})$ is the quantity that we are looking for.

Let $u$ be a node with $d$ successors $v_1,\ldots,v_d$.
Then $F_u$ is $c(u)$ plus a 
mixture
of the $d$ random variables $
 e_{uv_j}/\pforw_{uv_j}
\cdot F_{v_j}$, each taken
with probability $\alpha_j=\pforw_{uv_j}$.
Let $\mu_j$ and $\var_j$ denote the mean and variance of $F_{v_j}$.
Then, by Lemma~\ref{lem:mixture}, the mixture has mean
\begin{displaymath}
  \mu_0 = \sum_{j=1}^d \pforw_{uv_j} \cdot
  \frac{e_{uv_j}}{\pforw_{uv_j}} 
      \cdot \mu_j
      = \sum_{j=1}^d e_{uv_j}\mu_j
\end{displaymath}
and variance
\begin{align}
   \nonumber
  \var(F_u) &=
               \sum_{j=1}^d \pforw_{uv_j} \cdot \left[\left(
      \frac{e_{uv_j}}{\pforw_{uv_j}}
               \right)^2\cdot \var_j
    + \left( \frac{e_{uv_j}}{\pforw_{uv_j}} 
      \mu_j-\mu_0\right)^2\right]
  \\&
  \nonumber 
               = 
      \sum_{j=1}^d \frac1{\pforw_{uv_j}}
      \left[ 
  e_{uv_j}
  ^2\var_j
  + \left(e_{uv_j} \mu_j-\pforw_{uv_j}\mu_0\right)^2\right]
\end{align}
To get
the expectation of $F_u$, we must not forget to add the cost $c(u)$:
\begin{equation}
  \label{eq:EFu}
  E[F_u] = 
  c(u)+\mu_0
= c(u)+
  \sum_{j=1}^d e_{uv_j}\mu_j
\end{equation}
For a node $u$ with no successors, the formulas specialize to
$\var(F_u)
=0$ and
$E[F_u]=c(u)$.

For the special case of uniform forward proportions $\pforw_{uv_j}=1/d$, we get\iftrue\ the following
formulas:\else :\fi
\begin{displaymath}
\var(F_u) =   
d \cdot\sum_{j=1}^d  \left[
  e_{uv_j}^2\var_j
  + \left(e_{uv_j} \mu_j
    -\frac{\mu_0}d\right)^2\right]
=
d \cdot\sum_{j=1}^d  
  e_{uv_j}^2\var_j
  +
  \sum_{1\le i < j \le d}
  \left(
    e_{uv_i} \mu_i-    e_{uv_j} \mu_j
      \right)^2
\end{displaymath}

\iflong\else
\subsection{Combinatorial interpretation}
\label{sec:inter}

We can give a combinatorial interpretation of the expectation $E[F_v]$ of
the random variable $F_v$ introduced above: It is equal
to the quantity
$f_v = \sum _{w\in V} \mathcal{P}(v\to w) c(w)$ that we defined
in \eqref{eq:fv} for the analysis of Team Sampling
(Lemma~\ref{lem:martingale.deterministic}).

Comparing
\eqref{eq:EFu}
with \eqref{eq:fu=}
and taking note of the fact that $\mu_j$ stands
for $E[F_{v_j}]$, we see that
$E[F_v]$ and $f_v$ satisfy the same recurrence relation.
If $v$ is a sink,
\eqref{eq:EFu}
and \eqref{eq:fu=}
specify the same values
$E[F_v]=f_v=c(v)$, which can be used as initial conditions. 
 Thus:

\begin{proposition}  \label{lem:fr}
  The expected value of $F_v$ equals $f_v$.
  \qed
\end{proposition}

\fi

\fi
\section{Importance Sampling}
\label{sec:IS}

How can we use the freedom to pick the successor
according to an arbitrary distribution in order to get a better estimate?
For a rooted tree, it is known what the ideal forward
proportions
$\pforw_{uv}$ are: They should be proportional to the cost of
the subtrees rooted at the children $v$ of node~$u$. If the
forward proportions are chosen in this way, the estimate $X$ will
always deliver the correct answer.

In practice,
one can 
choose the probability $\pforw_{uv}$ of going from $u$ to $v$
 based on the properties of $v$ in the
graph, e.g., the number of successors or predecessors, or on
characteristics
of the combinatorial objects that the node $v$ represents.
We are entering the realm of heuristic arguments, and the
success of any particular rule depends on the type of problem.

We have tried the simple rule to choose the probabilities proportional
to the outdegrees of the successor nodes~$v$. (If the outdegree is 0,
we have taken 1 instead, to avoid zero probabilities.)
Nodes with large outdegree are more likely to reach many nodes.
For the permutation DAG, the
outdegree (and hence indegree) of a successor $v$ differs by at most
$\pm1$ from $u$, and hence the
effect of the outdegree tends to amplify.
This is why it seems to make sense to choose the outdegree as a criterion.

\begin{table}[htb]
  \centering
  \par\noindent\begin{tabular}{|c|r@{}lll|r@{}lll|}
\hline
\multicolumn{9}{|c|}{Permutations}\\
\hline
&\multicolumn4{c|}{$n=8$}&\multicolumn4{c|}{$n=9$}\\
$m$ & mean &\,$\pm$\,std.\,dev.&\multicolumn 1 c{max.}&\multicolumn 1 {c|}{$\hat\chi$}& mean &\,$\pm$\,std.\,dev.&\multicolumn 1 c{max.}&\multicolumn 1 {c|}{$\hat\chi$}\\
\hline\vbox to 11pt{}%
$10^1$&$10{,}564.8$&$\,\pm\,11{,}944.9$& $3.32{\times}10^{4}$ & $0.31394$&$234{,}362.6$&$\,\pm\,4.41{\times}10^{5}$& $1.28{\times}10^{6}$ & $0.54434$  \\
$10^2$&$15{,}677.9$&$\,\pm\,28{,}749.5$& $1.85{\times}10^{5}$ & $0.11808$&$315{,}348.3$&$\,\pm\,1.03{\times}10^{6}$& $6.34{\times}10^{6}$ & $0.20111$  \\
$10^3$&$47{,}912.4$&$\,\pm\,4.33{\times}10^{5}$& $1.23{\times}10^{7}$ & $0.25653$&$281{,}586.0$&$\,\pm\,1.10{\times}10^{6}$& $1.40{\times}10^{7}$ & $0.04980$  \\
$10^4$&$39{,}265.0$&$\,\pm\,2.24{\times}10^{5}$& $1.23{\times}10^{7}$ & $0.03130$&$312{,}165.1$&$\,\pm\,3.48{\times}10^{6}$& $2.36{\times}10^{8}$ & $0.07569$  \\
$10^5$&$40{,}970.8$&$\,\pm\,3.82{\times}10^{5}$& $6.63{\times}10^{7}$ & $0.01619$&$411{,}450.9$&$\,\pm\,2.05{\times}10^{7}$& $6.05{\times}10^{9}$ & $0.14698$  \\
$10^6$&$39{,}777.6$&$\,\pm\,4.31{\times}10^{5}$& $2.35{\times}10^{8}$ & $0.00592$&$360{,}495.0$&$\,\pm\,1.13{\times}10^{7}$& $6.05{\times}10^{9}$ & $0.01678$  \\
$10^7$&$40{,}289.8$&$\,\pm\,5.41{\times}10^{5}$& $7.25{\times}10^{8}$ & $0.00180$&$372{,}028.6$&$\,\pm\,4.02{\times}10^{7}$& $1.01{\times}10^{11}$ & $0.02708$  \\
\hline
\vbox to 11pt{}&$40{,}320\phantom{.0}$&$\,\pm\,\rlap{$6.316{\times}10^{5}$\ \ \ \#visited = 29}$&&&$362{,}880\phantom{.0}$&$\,\pm\,\rlap{$4.665{\times}10^{7}$\ \ \ \#visited = 37}$&&\\
\hline
\end{tabular}

  \caption{Algorithm P for permutations with Importance Sampling
    by outdegree}
  \label{tab:PERM-IS}
\end{table}

The results are shown in
\autoref{tab:PERM-IS}.
We see that
the theoretical standard deviation goes down by a factor between
2.6 and 3.5
in comparison with
\autoref{tab:P-PERM}.
In terms of empirical behavior, the comparison is not clear.
The experiments are
haunted by the occasional occurrences of rare events, such as unusually early appearances of high
estimates like $6.05 \times 10^9$ and
$1.01 \times 10^{11}$ for $n=9$.

We tried the same method also for Klee--Minty cubes, with similarly
inconclusive results.
\autoref{sec:importance} reports these results, as well as some further
experiments of only theoretical value, where the size of the reachable node set
is taken into account for the forward proportions.

\iflong
In another setting, namely the enumeration tree for pseudoline
arrangements,
it has been observed that the degrees of nodes and their children are
strongly correlated,
see \cite[\S9, 
Fig.~14]
{r-nert-25},
 suggesting that importance sampling should be helpful.
It indicates that
forward proportions chosen 
proportional to the outdegrees
would underestimate the true
size 
of the subtrees, 
and thus the
influence of the outdegrees on the probabilities should be even stronger than just proportional.
However,
an experimental evaluation has not been done yet.
A node in some deeper level of these enumeration trees has typically several hundred to thousands of successors.
\fi

\section{Employing a team of agents}
\label{sec:team}
\subsection{The Team Sampling family of algorithms (TS)}
\label{sec:TS}

Instead of \iflong making \fi many independent repetitions of following a single path with Algorithm~P,
it is better to let
several ``tokens'' or ``cooperative agents'' move through the graph
simultaneously.
\emph{In this section, we assume that the arc weights $e_{uv}$ are positive}.

We propose the following general scheme.
The algorithm maintains a list $A$ of \emph{active nodes}, and for each
element $v\in A$, a positive weight $W[v]$.

\medskip\par
\begin{centering}
\fboxsep=7pt \noindent\fbox
{\hskip -1.em \begin{minipage}[c]{0.99\textwidth}
\begin{enumerate}
\item[\rlap{\setbox0=\hbox{\bf 1.}\hskip -\wd0 \bf Algorithm TS: Team Sampling}]
\item[\bf 1.]
  [Initialize] $A:=\{s_0
  \}$; \
  $W[s_0
  ] := 1$; \ $X := 0$.
\item[\bf 2.] While $A$ is nonempty, perform either of the following two steps:
\begin{enumerate}
\item [\bf 2E.] [Expand successors]
  Choose an arbitrary node $u$ of $A$,
  and remove it 
  from $A$.
  
  Set $X := X + W[u]\cdot 
  c(u)$. /\!/ Augment $X$
  

  For each successor $v$ of $u$:\\
\null  \qquad
If $v$ is not 
present in $A$,
insert $v$ into $A$ and clear its weight: $W[v] := 0$.
\\\null\qquad Then,
increase its weight to
$W[v] := W[v] +
W[u] e_{uv}
$.

\item [\bf 2C.] [
  Condense]
   Choose an {arbitrary} subset $S$ of $A$ with at least 2
   elements.
   
   Pick a random node $u$ from $S$, with probabilities proportional to
the weights~$W[u]$.
   
   Replace $S$ in $A$ by the single element $u$, and set $W[u]$ to the
   sum of the weights in $S$.

\end{enumerate}  
 \item [\bf 3.] [Terminate] In the end, return $X$.
\end{enumerate}
\end{minipage}}\par\end{centering}\goodbreak\par\bigskip

The choice between
 Step 2E and Step 2C, and the
choice of $u$ in Step 2E or of $S$ in Step 2C can be done according to some
rule, depending on any information that has been stored or that is
currently at hand, or even randomly. The node weights can be taken into account.
We prove
in \autoref{sec:TS-correctness}
that, whatever choices the algorithm makes, it returns  an unbiased
estimate $X$ of the quantity $c_e(G)$ that was defined
in \autoref
{th:unbiased-paths}.

Algorithm P can be seen as a special case of this algorithm where Step~2E is always directly followed
by Step~2C for the whole set $A$, i.\,e., the successors of the current
node.
(To make this correspondence precise,
the forward proportions $\pforw_{uv}$ must be chosen proportional to
$e_{uv}$:
When Step~2C is applied to the successors $v$ of the current node~$u$,
whose weights are $W[u]\,e_{uv}$, it then picks $v$ with
probability $\pforw_{uv}=e_{uv}/\sum_w e_{uw}$, and the new weight is $W[u]\sum_w
e_{uw} = W[u]\,e_{uv}/\pforw_{uv}$, in accordance with the
update formula (U$_e$).)
%

In contrast to Path Sampling,  Algorithm TS as formulated
does not allow
 an arbitrary probability distribution for selecting
 among the successors.
\autoref{sec:mix} will show how 
an arbitrary nonuniform selection
can be  accommodated by
 generalizing
the condensation step~2C.

It can happen that
a node $v$ is removed from $A$ and later reinserted.
For example, in the graph of
\autoref{fig:example},
after initially expanding $v_1$ into the set $A=\{v_2,v_3\}$,
the algorithm could first expand~$v_3$, and later $v_2$, reinserting
$v_3$ into $A$.
This is obviously a waste of effort, and it can be avoided by choosing the nodes in Step~2E in some
topological order of the DAG.
In many applications, it is easy to obey this rule.
For example, the permutation DAG is graded: All paths from the source
to a node~$v$ have the same length.
More generally, we may have some rank function $r\colon V\to \mathbb
Z$
or $r\colon V\to \mathbb R$ such that $r(v)>r(u)$ for all arcs $uv$.
When we have the graph of a polytope, oriented by some generic linear objective
function~$f$, then $f$ 
can serve as
a rank function for the nodes.
In other instances, it may not be straightforward to find out if a
node $v$ is reachable from a node~$u$. In any case,
 the correctness of the algorithm does not
depend on any proper order in which the nodes are chosen.

%

\iflong

\paragraph{Keeping multiple copies of the same node.}
In the above 
formulation of Step 2E, we check the presence of a node
$v$ in $A$ before trying to insert it, but this is actually not 
necessary.
We can also insert a fresh copy of $v$ into $A$ with a separate weight
$W[v]$,
(which is initialized to~0 as usual).
Accordingly,
we have to treat $A$ as a
\emph{multiset}.
When such a duplication is detected, it is of course
expedient to condense multiple copies of the same node $v$, as a
special case of Step~2C; but again, the correctness of the algorithm does not
depend on this.

\fi

\medbreak

The general scheme of Team Sampling offers a lot of freedom.
We describe and test two strategies that specialize this scheme,
\emph{Bucket Sampling} and \emph{Crowd Sampling}, both
inspired by methods that have been 
proposed
for trees in the literature.

\enlargethispage{-1em}

\subsection{Bucket Sampling (Algorithm BS)}
\label{sec:bucket}
This strategy is modeled after the
algorithm of
Pang C. Chen
\cite{Chen-sampling-92} from 1992, which was proposed for trees under the name
\emph{heuristic sampling}; Knuth 
called it \emph{Chen subset sampling}
\cite[Section 7.2.2.9, 
Algorithm~C%
]{kn4c}.
Nodes with similar
characteristics are lumped together in a group~$S$ that is to 
be
condensed.
The idea is that,
if the groups are formed in a sensible way, any element $u\in S$ is as good a representative of $S$ as
any other element.

 Formally,
 the algorithm uses a \emph{heuristic function}
$h\colon V\to U$ from the nodes $V$ to some universe~$U$, which is
used
to identify 
the
 \emph{buckets} to which the nodes are assigned.
Ideally, the criterion 
by which nodes are grouped would be that the nodes
\iflong
have similar subtrees (when $G$ is a tree) or similar sets of successors
(when $G$ is a DAG).
\else
have similar sets of successors.
\fi
However,
these characteristics are usually not available or expensive to compute, so
one has to 
settle for some easier ``heuristic'' characteristics.
In our experiments, we have chosen the outdegree as the heuristic function.

In the original algorithm of
 Chen
\cite{Chen-sampling-92}, 
a node and its ancestor node cannot be in one bucket, and
what we call a bucket is called a \emph{stratum}.
\iflong
This requirement, together with a proper order of processing the strata,
ensures that the algorithm visits at most one node per stratum.
With our relaxed notion
of heuristic function, the stratum terminology is less appropriate,
and we have chosen the term bucket instead.
\fi

Depending on the character of the universe~$U$, the list of buckets can be
organized in an array, a hash table, or a priority queue.
\iflong\par
Conceptually, each bucket could store a list containing the (multi-)set of nodes $v$ with
the same value $h(v)$. \emph{The algorithm repeatedly chooses some nonempty
bucket, condenses it, and expands the chosen representative.}
\fi
\iflong

However, the following considerations, which were already made
for the original heuristic sampling method of~\cite{Chen-sampling-92},
show that one only needs to store at most one node per bucket.
Suppose we
condense a set $S=\{u_1,u_2,u_3\}$, resulting in a
representative  $u'\in S
$, and later condense
the set $\{u',u_4,u_5,u_6\}$.
It can be easily checked that the total effect is the same as
condensing the set 
$\{u_1,u_2,u_3,u_4,u_5,u_6\}$ in one go.
Thus, if we know that $u_1,u_2,u_3$ will eventually be condensed (possibly
together with more elements that are not yet known), it makes sense to
condense them right away, in order to save space.

Therefore, we store at most one node per bucket.  In the expansion
step, we generate successor nodes and insert them into their buckets.
Whenever a node is inserted into a bucket that is already
occupied, we immediately condense the set of these two nodes.
(This procedure corresponds to the \emph{reservoir sampling} method of selecting
a sample from a data stream.)
\iflong\else\looseness -1 \fi

The resulting algorithm with the modified main loop is shown here:
\bigskip
\medskip\par
\begin{centering}
\fboxsep=7pt \noindent\fbox
{\hskip -1.em \begin{minipage}[c]{0.96\textwidth}
\begin{enumerate}
\item[\rlap{\setbox0=\hbox{\bf 1.}\hskip -\wd0 \bf
    Algorithm BS: Bucket Sampling}]
\item[\bf 1.]
  [Initialize] $A:=\{s_0
  \}$; \
  $W[s_0
  ] := 1$; \ $X := 0$.
\item[\bf 2.] While $A$ is nonempty, perform the following step:
\begin{description}
\item [\bf BS-2E.] [Expand successors]
  Choose an arbitrary node $u$ of $A$,
  and remove it 
  from $A$.
  
  Set $X := X + W[u]\cdot 
  c(u)$. /\!/ Augment $X$
  

  For each successor $v$ of $u$:\\
\null  \qquad
If $A$ contains no node $v'$ with $h(v)=h(v')$,
\\
\null  \qquad  \qquad
Insert $v$ into $A$ and
set $W[v] := 
W[u] e_{uv}
$.
\\\null  \qquad
Otherwise, i.e., if $A$ contains a node $v'$ with $h(v)=h(v')$,
\\\null  \qquad \qquad
 [Condense $v$ and $v'$]
 Compute $\bar W
 := W[v'] +
W[u] e_{uv}
$.
 \\\null  \qquad \qquad
With probability
$W[v']/\bar W$,
set
 $W[v'] := \bar W
$ (and ignore $v$).
\\\null  \qquad \qquad
With the complementary probability, 
\\\null  \qquad \qquad \qquad
replace $v'$ by $v$ in $A$,
and set $W[v] := \bar W$.
\end{description}  
 \item [\bf 3.] [Terminate] In the end, return $X$.
\end{enumerate}
\end{minipage}}\par\end{centering}\goodbreak\par
\bigskip

\else
As observed already
for the original heuristic sampling method of~\cite{Chen-sampling-92},
it is sufficient to store at most one node per bucket.  In the expansion
step, we generate successor nodes and insert them into their buckets.
Whenever a node is inserted into a bucket that is already
occupied, we immediately condense the set of these two nodes.
\iflong
(This procedure corresponds to the \emph{reservoir sampling} method of selecting
a sample from a data stream.)
\fi
\iflong\else\looseness -1 \fi

\fi

\enlargethispage{-1em}
\paragraph{Experiments.}

We have tried this approach for the permutation DAG and the Klee--Minty
cubes.
As a heuristic function, we have taken the outdegree.
In each step, we have selected the nonempty bucket with the largest
outdegree and expanded the node stored in this bucket.
The rationale behind this choice is that nodes closer to the root tend to have
larger outdegree.
\iflong
For the permutation DAG, it would have been natural to use
the level  (rank) as
additional information, but we have ignored this for
simplicity.\footnote
{Some preliminary experiments (not reported here) suggest that
  bucketing by level \emph{and} outdegree gives a substantial
  improvement:
It reduces the standard deviation by a factor of 4 to 6, and it markedly reduces the number of visited nodes.}
\fi

The results are shown in
Tables~\ref{tab:PERM-BS}
and~\ref{tab:KM-BS}.
The number of iterations (reported as ``\#visited'') is larger than
for Path Sampling, and in addition, each visit
incurs more work
than for Path Sampling; therefore we went only up to $m=10^6$
repetitions.
One can see that Bucket Sampling with our very simple bucketing strategy
improves
the estimate.
Bucket Sampling reduces the empirical standard deviation by a factor
of about 2 to 12 in comparison with Algorithm P, but it remains larger
than the mean.

We also report as ``\#generated'' the total number of successor nodes
that are generated and inserted into buckets; this quantity is
representative of the total work of the algorithm.
\iflong
Our digraphs are regular graphs, of degree $n-1$ for the permutation
DAGs, and of degree $n$ for the Klee--Minty cubes. 
 On average, half of the incident arcs are outgoing arcs.
Thus the number of generated nodes should be roughly the number of visited
nodes times $(n-1)/2$ or $n/2$, respectively.
The actual number of generated nodes
(``\#generated'')
is just slightly smaller than what one
would
expect from this calculation.
\fi

\begin{table}[htb]
  \centering
  \par\noindent\begin{tabular}{|c|r@{}lll|r@{}lll|}
\hline
\multicolumn{9}{|c|}{Permutations}\\
\hline
&\multicolumn4{c|}{$n=8$}&\multicolumn4{c|}{$n=9$}\\
$m$ & mean &\,$\pm$\,std.\,dev.&\multicolumn 1 c{max.}&\multicolumn 1 {c|}{$\hat\chi$}& mean &\,$\pm$\,std.\,dev.&\multicolumn 1 c{max.}&\multicolumn 1 {c|}{$\hat\chi$}\\
\hline\vbox to 11pt{}%
$10^1$&$41{,}711.3$&$\,\pm\,68{,}536.3$& $2.12{\times}10^{5}$ & $0.50920$&$120{,}931.2$&$\,\pm\,1.24{\times}10^{5}$& $4.52{\times}10^{5}$ & $0.37377$  \\
$10^2$&$39{,}047.5$&$\,\pm\,1.40{\times}10^{5}$& $1.36{\times}10^{6}$ & $0.34866$&$258{,}487.9$&$\,\pm\,9.36{\times}10^{5}$& $8.60{\times}10^{6}$ & $0.33284$  \\
$10^3$&$34{,}062.1$&$\,\pm\,89{,}284.0$& $1.73{\times}10^{6}$ & $0.05073$&$398{,}791.7$&$\,\pm\,3.63{\times}10^{6}$& $1.08{\times}10^{8}$ & $0.27097$  \\
$10^4$&$38{,}986.7$&$\,\pm\,1.34{\times}10^{5}$& $5.97{\times}10^{6}$ & $0.01532$&$437{,}469.8$&$\,\pm\,1.12{\times}10^{7}$& $1.09{\times}10^{9}$ & $0.25028$  \\
$10^5$&$39{,}853.5$&$\,\pm\,2.02{\times}10^{5}$& $4.61{\times}10^{7}$ & $0.01157$&$357{,}426.5$&$\,\pm\,4.90{\times}10^{6}$& $1.09{\times}10^{9}$ & $0.03063$  \\
$10^6$&$40{,}025.5$&$\,\pm\,2.55{\times}10^{5}$& $1.28{\times}10^{8}$ & $0.00321$&$360{,}984.5$&$\,\pm\,5.31{\times}10^{6}$& $2.67{\times}10^{9}$ & $0.00740$  \\
\hline
\vbox to 11pt{}&$40{,}320\phantom{.0}$&$\rlap{}$&&&$362{,}880\phantom{.0}$&$\rlap{}$&&\\
&\multicolumn4{l|}{\#visited: avg\ = 84.4, max\ = 277}&\multicolumn4{l|}{\#visited: avg\ = 180.7, max\ = 763}\\
&\multicolumn4{l|}{\#generated: avg\ = 290.6}&\multicolumn4{l|}{\#generated: avg\ = 690.5}\\
\hline
\end{tabular}

  \caption{Bucket Sampling for the permutation DAG.}
  \label{tab:PERM-BS}
\end{table}

\begin{table}[htb]
  \centering
  \par\noindent\begin{tabular}{|c|r@{}lll|r@{}lll|}
\hline
\multicolumn{9}{|c|}{Klee--Minty cubes}\\
\hline
&\multicolumn4{c|}{$n=10$}&\multicolumn4{c|}{$n=12$}\\
$m$ & mean &\,$\pm$\,std.\,dev.&\multicolumn 1 c{max.}&\multicolumn 1 {c|}{$\hat\chi$}& mean &\,$\pm$\,std.\,dev.&\multicolumn 1 c{max.}&\multicolumn 1 {c|}{$\hat\chi$}\\
\hline\vbox to 11pt{}%
$10^1$&$426.2$&$\,\pm\,297.3$& $1.01{\times}10^{3}$ & $0.23671$&$2{,}912.5$&$\,\pm\,3{,}257.5$& $9.80{\times}10^{3}$ & $0.33636$  \\
$10^2$&$1{,}002.6$&$\,\pm\,2{,}417.0$& $2.08{\times}10^{4}$ & $0.20746$&$2{,}805.3$&$\,\pm\,3{,}774.8$& $2.08{\times}10^{4}$ & $0.07424$  \\
$10^3$&$975.0$&$\,\pm\,1{,}852.5$& $2.49{\times}10^{4}$ & $0.02554$&$4{,}418.8$&$\,\pm\,20{,}658.1$& $4.78{\times}10^{5}$ & $0.10808$  \\
$10^4$&$970.3$&$\,\pm\,2{,}375.3$& $8.89{\times}10^{4}$ & $0.00916$&$3{,}838.8$&$\,\pm\,17{,}371.0$& $1.21{\times}10^{6}$ & $0.03150$  \\
$10^5$&$1{,}035.6$&$\,\pm\,8{,}173.8$& $2.13{\times}10^{6}$ & $0.02057$&$3{,}982.0$&$\,\pm\,27{,}084.6$& $4.81{\times}10^{6}$ & $0.01209$  \\
$10^6$&$1{,}024.6$&$\,\pm\,4{,}759.7$& $2.13{\times}10^{6}$ & $0.00208$&$4{,}106.5$&$\,\pm\,50{,}004.1$& $3.63{\times}10^{7}$ & $0.00884$  \\
\hline
\vbox to 11pt{}&$1{,}024\phantom{.0}$&$\rlap{}$&&&$4{,}096\phantom{.0}$&$\rlap{}$&&\\
&\multicolumn4{l|}{\#visited: avg\ = 193.2, max\ = 1105}&\multicolumn4{l|}{\#visited: avg\ = 542.4, max\ = 4291}\\
&\multicolumn4{l|}{\#generated: avg\ = 937.0}&\multicolumn4{l|}{\#generated: avg\ = 3092.3}\\
\hline
\end{tabular}

  \caption{Bucket Sampling for Klee--Minty cubes.}
  \label{tab:KM-BS}
\end{table}

\enlargethispage{-1em}

\subsection{Crowd Sampling (Algorithm CS)}
\label{sec:crowd}
This strategy is inspired by the
\emph{stochastic enumeration} (SE)
algorithm, which was proposed
in 2013 by
Reuven Y. Rubinstein~\cite{Rubinstein-2013}
to count the number of nodes of
a tree.
The algorithm has a budget parameter $B$, imposing a limit on the number
of nodes that we want to keep for further expansion. The algorithm
consists of an alternating sequence of expansion and condensation
phases. The main loop is as follows.

\goodbreak

\begin{description}
\item [\bf CS-2E] [Expansion]\\* Expand every node of the current set $A$.
  \\Let $A'$ denote the resulting multi-set (or set, after identifying duplicates).
\item [\bf CS-2C] [Condensation]\\
  If $|A'|\le B$, let $A:= A'$.
  \\
  Otherwise, arbitrarily 
  partition $A'$ into $B$ sets $S_1\cup \cdots \cup S_B$ of
  size
  $\lfloor \frac
  {|A'|}B
  \rfloor$
  or  $\lceil \frac{|A'|}B \rceil$.
  Condense each set $S_i$ to a single element, and let $A$ be the set
  of
  these $B$ elements.
\end{description}

We call this method \emph{Crowd Sampling}, because the samples are
treated more like an unorganized crowd
than a cooperating team.

The purpose of Step~CS-2C is to reduce the size of the active set $A'$ to $B$.
The method we have proposed is a straightforward way to achieve this.
In our implementation, we permute the set $A'$ randomly and then
partition it into sets of the required size.
One can think of simpler methods, like a round-robin distribution in
the order in which the elements happen to be stored in $A'$,
or more sophisticated methods, taking into account
the weights, for example.
Unbiasedness does not depend on the way how the partition is carried out.

Tables~\ref{tab:PERM-CS}
and~\ref{tab:KM-CS} show the results of this method with a budget of
$B=10$ elements. We also report
the number of nodes that were saved
by condensing multiple copies of the same node after the expansion step
(``\#saved'').
Compared to our Bucket Sampling experiments
(Tables~\ref{tab:PERM-BS}--\ref{tab:KM-BS}), we see that the standard
deviation is further reduced, at the expense of a larger number of
visited nodes. The  number of
visited nodes can be controlled by adjusting the parameter~$B$.

\iflong
There is another partial backtracking method, proposed by Purdom~\cite{purdom-78}, which
pursues a similar goal of examining every level of a tree
at about the same frequency. It is discussed in \autoref{sec:purdom}.
\fi

\begin{table}[htb]
  \centering
  \par\noindent\begin{tabular}{|c|r@{}lll|r@{}lll|}
\hline
\multicolumn{9}{|c|}{Permutations}\\
\hline
&\multicolumn4{c|}{$n=8$}&\multicolumn4{c|}{$n=9$}\\
$m$ & mean &\,$\pm$\,std.\,dev.&\multicolumn 1 c{max.}&\multicolumn 1 {c|}{$\hat\chi$}& mean &\,$\pm$\,std.\,dev.&\multicolumn 1 c{max.}&\multicolumn 1 {c|}{$\hat\chi$}\\
\hline\vbox to 11pt{}%
$10^1$&$51{,}102.7$&$\,\pm\,43{,}393.4$& $1.56{\times}10^{5}$ & $0.30491$&$199{,}951.9$&$\,\pm\,1.45{\times}10^{5}$& $4.06{\times}10^{5}$ & $0.20327$  \\
$10^2$&$43{,}250.2$&$\,\pm\,33{,}729.6$& $1.56{\times}10^{5}$ & $0.03603$&$370{,}850.2$&$\,\pm\,3.97{\times}10^{5}$& $2.05{\times}10^{6}$ & $0.05538$  \\
$10^3$&$43{,}252.0$&$\,\pm\,39{,}430.2$& $3.70{\times}10^{5}$ & $0.00855$&$317{,}016.9$&$\,\pm\,3.77{\times}10^{5}$& $3.08{\times}10^{6}$ & $0.00971$  \\
$10^4$&$40{,}068.4$&$\,\pm\,35{,}080.8$& $7.40{\times}10^{5}$ & $0.00185$&$359{,}232.3$&$\,\pm\,5.80{\times}10^{5}$& $1.84{\times}10^{7}$ & $0.00513$  \\
$10^5$&$40{,}222.5$&$\,\pm\,35{,}171.4$& $8.77{\times}10^{5}$ & $0.00022$&$362{,}323.6$&$\,\pm\,5.76{\times}10^{5}$& $2.37{\times}10^{7}$ & $0.00066$  \\
$10^6$&$40{,}259.1$&$\,\pm\,35{,}177.0$& $1.37{\times}10^{6}$ & $0.00003$&$361{,}797.9$&$\,\pm\,5.81{\times}10^{5}$& $7.45{\times}10^{7}$ & $0.00021$  \\
\hline
\vbox to 11pt{}&$40{,}320\phantom{.0}$&$\rlap{}$&&&$362{,}880\phantom{.0}$&$\rlap{}$&&\\
&\multicolumn4{l|}{\#visited: avg\ = 261.1, max\ = 266}&\multicolumn4{l|}{\#visited: avg\ = 341.8, max\ = 348}\\
&\multicolumn4{l|}{\#generated: avg\ = 893.4}&\multicolumn4{l|}{\#generated: avg\ = 1329.4}\\
&\multicolumn4{l|}{\#saved: avg\ = 178.7 = 20.0\%}&\multicolumn4{l|}{\#saved: avg\ = 240.4 = 18.1\%}\\
\hline
\end{tabular}

  \caption{Crowd Sampling for the permutation DAG with a budget $B=10$.}
  \label{tab:PERM-CS}
\end{table}

\begin{table}[htb]
  \centering
  \par\noindent\begin{tabular}{|c|r@{}lll|r@{}lll|}
\hline
\multicolumn{9}{|c|}{Klee--Minty cubes}\\
\hline
&\multicolumn4{c|}{$n=10$}&\multicolumn4{c|}{$n=12$}\\
$m$ & mean &\,$\pm$\,std.\,dev.&\multicolumn 1 c{max.}&\multicolumn 1 {c|}{$\hat\chi$}& mean &\,$\pm$\,std.\,dev.&\multicolumn 1 c{max.}&\multicolumn 1 {c|}{$\hat\chi$}\\
\hline\vbox to 11pt{}%
$10^1$&$698.9$&$\,\pm\,362.5$& $1.49{\times}10^{3}$ & $0.21311$&$1{,}747.4$&$\,\pm\,1{,}241.1$& $4.00{\times}10^{3}$ & $0.22862$  \\
$10^2$&$1{,}010.6$&$\,\pm\,1{,}053.9$& $6.45{\times}10^{3}$ & $0.06381$&$3{,}533.0$&$\,\pm\,6{,}814.3$& $5.70{\times}10^{4}$ & $0.16136$  \\
$10^3$&$1{,}006.2$&$\,\pm\,1{,}586.2$& $2.99{\times}10^{4}$ & $0.02973$&$3{,}861.8$&$\,\pm\,11{,}011.9$& $2.55{\times}10^{5}$ & $0.06613$  \\
$10^4$&$1{,}023.7$&$\,\pm\,1{,}559.8$& $5.64{\times}10^{4}$ & $0.00551$&$4{,}172.8$&$\,\pm\,12{,}451.9$& $5.04{\times}10^{5}$ & $0.01208$  \\
$10^5$&$1{,}019.5$&$\,\pm\,1{,}446.9$& $7.24{\times}10^{4}$ & $0.00071$&$4{,}107.1$&$\,\pm\,11{,}862.7$& $8.52{\times}10^{5}$ & $0.00207$  \\
$10^6$&$1{,}022.6$&$\,\pm\,1{,}490.0$& $2.40{\times}10^{5}$ & $0.00023$&$4{,}104.1$&$\,\pm\,13{,}739.7$& $3.46{\times}10^{6}$ & $0.00084$  \\
\hline
\vbox to 11pt{}&$1{,}024\phantom{.0}$&$\rlap{}$&&&$4{,}096\phantom{.0}$&$\rlap{}$&&\\
&\multicolumn4{l|}{\#visited: avg\ = 734.0, max\ = 1381}&\multicolumn4{l|}{\#visited: avg\ = 971.0, max\ = 1880}\\
&\multicolumn4{l|}{\#generated: avg\ = 2559.3}&\multicolumn4{l|}{\#generated: avg\ = 4012.0}\\
&\multicolumn4{l|}{\#saved: avg\ = 408.3 = 16.0\%}&\multicolumn4{l|}{\#saved: avg\ = 433.3 = 10.8\%}\\
\hline
\end{tabular}

  \caption{Crowd Sampling for Klee--Minty cubes with a budget $B=10$.}
  \label{tab:KM-CS}
\end{table}

\subsection{Correctness of the Team Sampling scheme}
\label{sec:TS-correctness}

\iflong
We use the notation
$\mathcal{P}(v\to w)$ for the sum of all path weights from $v$ to $w$:
\begin{displaymath}
  \mathcal{P}(v\to w)=     \sum_{\textup{paths $P$ from
        $v$ to $w$}}
    e(P)
\end{displaymath}
\else
We use the notation
$
  \mathcal{P}(v\to w)=     \sum_{\textup{paths $P$ from
        $v$ to $w$}}
    e(P)
$ for the sum of path weights. 
\fi

\begin{lemma}
  \label{lem:martingale.deterministic}
Suppose that at some stage in Algorithm TS, we have a current
estimate
$X^{\mathrm{current}}$ and a current active set $A$ with weights $W[v]$
for $v\in A$.
If we
perform only expansion steps (Step~2E)
 until the end of the algorithm,
 then the
 value of the 
 estimate
 $X$
 at termination can be predicted deterministically, and it is
\begin{equation}
  \label{eq:outcome-det}
  X^{\mathrm{current}} + \sum_{v\in A}
W[v]f_v,
\end{equation}
for some factors $f_v$ 
that depend only on the graph $G$, the costs $c(v)$,  and
the arc weights~$e_{vw}$\textup:%
\begin{equation}
  \label{eq:fv}
  f_v := \sum _{w\in V} \mathcal{P}(v\to w) c(w)
\end{equation}
In particular, an expansion step does not change the
 value 
 of~\eqref{eq:outcome-det}.
\end{lemma}
\begin{proof}
  It is clear that the last sentence implies the whole lemma, since,
  at termination, $A=\emptyset$, and 
  $ 
  X^{\mathrm{current}}$ is the final value of~$X$.

  Consider the expansion of a node $u\in A$, and
  let $v_1,\ldots,v_d$  be the $d^+(u)$ successors of~$u$.
  It follows from the definition that $f_u$ satisfies the 
  relation
\begin{equation}
  \label{eq:fu=}
  f_u = c(u) + \sum _{i=1}^{d^+(u)} e_{uv_i}f_{v_i}.
\end{equation}
Expanding node $u\in A$
adds $c(u)W[u]$ to $X^{\mathrm{current}}$,
and it adds $e_{uv_i} W[u]$ to each $W[v_i]$,
so the expression
\eqref{eq:outcome-det}
increases by $W[u](c(u)+ \sum _{i=1}^{d^+(u)} e_{uv_i} f_{v_i})$.
On the other hand, $u$ is removed from~$A$,
causing a decrease by 
$W[u]f_u$. By \eqref{eq:fu=} these two changes cancel out.
\end{proof}

\iflong
We now show that
the value
 of~\eqref{eq:outcome-det} is left invariant also by a condensation
 step.
In contrast to the expansion step
(Lemma~\ref{lem:martingale.deterministic}), this holds only in expectation:
\fi

\begin{lemma}
  \label{lem:invariant-C}
  After a condensation step, 
  the expected
value of
$\sum_{v\in A}
W[v]f_v$ is
\iflong
the same as before.
\else
unchanged.
\fi
\end{lemma}
\begin{proof}
  Condensing
  replaces a set $S$ by one of its members $u\in S$,
where $u$ is chosen with probability $W[u]/\sum_{v\in S}W[v]$.
Let $\hat A=(A\setminus S)\cup\{u\}$ be the new active set, and let $W'$ denote the new weights.
The new weight of $u$ is $W'[u] =\sum_{v\in S}W[v]$.
The weights in $A\setminus S$ are unchanged.
Thus we get
\iflong  \else the claimed relation\fi
\begin{align}
  \label{eq:unbi1}
  E\Biggl[ \sum_{v\in \hat A} W'[v]f_v\Biggr]
  &=
    \sum_{v\in A\setminus S} W[v]f_v  +
    \sum_{u\in S}\Bigl(
    \Pr[\text{$u$ is chosen}] \times W'[u] f_u \Bigr)
  \\
  \label{eq:unbi2}
  &=
  \sum_{v\in A\setminus S} W[v]f_v  +
    \sum_{u\in S}\left(
    \frac {W[u]}
    {\sum\limits_{v\in S}W[v]} \times \sum_{v\in S}W[v]
    f_u                          \right)
  \\  \label{eq:unbi3}
  &=
  \sum_{v\in A\setminus S} W[v]f_v  +
  \sum_{u\in S} W[u]f_u  
=
    \sum_{v\in A} W[v]f_v
\iflong    ,\else.\qedhere\fi
\end{align}
\iflong as claimed. \fi
\end{proof}

\begin{theorem}
  Algorithm TS produces an estimate with
  the same expectation as Algorithm~P: 
    \begin{displaymath}
    E[X] =   c_e(G) =  \sum_{v\in V} \mathcal{P}(s_0\to v)c(v)
  \end{displaymath}
\end{theorem}
\begin{proof}
Lemmas \ref{lem:martingale.deterministic} and \ref
{lem:invariant-C}
together
imply that, at any time during the algorithm,
\begin{displaymath}
  E[X^{\mathrm{final}}]
  = X^{\mathrm{current}} + \sum_{v\in A}
W[v]f_v.
\end{displaymath}
If  we apply this to the starting situation of Algorithm
  TS, with $A
    =\{s_0
  \}$,
  $W[s_0
  ] = 1$, and $X^{\mathrm{current}}=0$,
  we get $E[X^{\mathrm{final}}]= f_{s_0}$, which equals the claimed quantity.
\end{proof}

\iflong

As a side result of this analysis, we get a combinatorial interpretation of
the expectation 
of
the random variables $F_v$ introduced
in \autoref{sec:variance} to compute the variance of the estimate~$X$.
%
Comparing
\eqref{eq:fu=}
with \eqref{eq:EFu}
and taking note of the fact that $\mu_j$ in \eqref{eq:EFu} stands
for $E[F_{v_j}]$, we see that
$f_v$
 and
$E[F_v]$
satisfy the same recurrence relation.
If $v$ is a sink,
\eqref{eq:fu=}
and \eqref{eq:EFu}
specify the same values
$
f_v
=
E[F_v]
=c(v)$, which can be used as initial conditions. 
 Thus:

\begin{proposition}  \label{lem:fr}
  The expected value of $F_v$ equals $f_v$.
  \qed
\end{proposition}

\fi

\subsection{Generalized Condensation}
\label{sec:mix}
\iflong
The Team Sampling algorithm shown 
in \autoref{sec:TS}
proposes a specific method for
the Condensation Step (choice of probabilities and weight update).
We will see that this is just one of many possibilities.

\fi

As mentioned, Crowd Sampling was modeled after the Stochastic
Enumeration algorithm of Rubinstein~\cite{Rubinstein-2013},
which works for trees. For a tree, the algorithm becomes very
simple: In the expansion step, all backward proportions $\qback_{uv}$ are~1, and hence the
children just get the same weights as their parents.
If the resulting set $A'$ of children is too big, it is
simply replaced by a random sample of size $B$, and the weight is
accordingly multiplied by $|A'|/B$. All nodes have the same weight.

Our Crowd Sampling method mimics the behavior of this algorithm only
when
$|A'|$ is a multiple of $B$
and the partition of $A'$ into equal-size subsets $S_i$ is random.
Otherwise, Crowd Sampling generates
nodes of different weights, introducing irregularities into the process.
In fact, for DAGs,
nodes of different weights
\iflong
appear
anyway, for another reason.
\else
necessarily appear.
\fi
For a set
with weights,
the meaning 
of
``a random sample of size $B$'' 
 is not obvious;
see the survey \cite{Efraimidis2015}
for two different potential interpretations.
%
%
Let us take an
{engineering approach} and formulate
what is required to ensure that the estimate remains unbiased.
It will turn out that
this actually leaves a lot of freedom in performing the condensation step.




Assume that the set $S$ contains $k$ nodes with weights
$W_1,W_2,\ldots,W_k$.
Suppose we decide that
Step 2C should randomly select the $i$-th node with probability $p_i$,
and assign weight $W'_i$ to it in case it is selected.
If we look at the derivation \thetag{\ref{eq:unbi1}--\iflong
  \ref{eq:unbi3}\else  \ref{eq:unbi2}\fi
} in 
the proof of Lemma~\ref{lem:invariant-C}, 
we see that
what we need for
unbiasedness is the equation
\begin{equation}
  \label{eq:unbiased}
  p_1W'_1 f_1 +
  \cdots +  p_kW'_k f_k =
  W_1f_1 + \cdots + W_kf_k.
\end{equation}
Since we don't know the values $f_i$, this equation should hold for
all $f_1,\ldots,f_k$. Comparing coefficients leads to $k$ equations
\begin{equation}
  \label{eq:necessary}
    p_iW'_i = W_i,
\end{equation}
for $i=1,\ldots,k$.  
Together with the equation
$\sum_{i=1}^k p_i=1$, this gives $k+1$ equations for $2k$ unknowns
$p_i,W'_i$, which is far from enough to define a unique solution.\footnote
{We may even permit
  $\sum p_i<1$,
  effectively 
  throwing away
  the sample
with probability  $1-\sum p_i$. In the single path variation (Algorithm~P),
the algorithm gets a chance to finish early and take a
rest. This comes, of course, at the expense of a larger fluctuation of
the
estimate.}

The basic condensation step of Algorithm TS chooses $W'_i = \bar W := W_1+\cdots +W_k$ and $p_i =
W_i/\bar W$,
but it is not 
clear why this choice should be preferable over
other choices.
In fact, we can select arbitrary probabilities $p_i$ (for example,
uniform probabilities) and set $W_i' :=
W_i/p_i$. In this way, we can reinstate the forward proportions
of Algorithm~P, which Algorithm TS in its original form could not
handle.
\iflong\else\looseness -1 \fi

\iflong
The generalized condensation step for Algorithm TS is thus as follows:

\begin{description}
\item [\bf 2C$'$.] [
  Condense]
   Choose an {arbitrary} subset $S=\{u_1,u_2,\ldots,u_k\}$ of $A$ with
   $k\ge2$ 
   elements.
\\
Choose positive probabilities $p_i$ with   
$
p_1+p_2+\cdots+p_k=1$.   
  \\ 
   Pick a random node $u_i$ from $S$ according to this distribution.
   \\
   Replace $S$ in $A$ by the single element $u_i$, and set $W[u_i]:=W[u_i]/p_i$.
   \end{description}

\fi
   
\iflong
\subsection{The variance for Team Sampling strategies}
\label{sec:var-general}

For particular strategies that specialize the Team Sampling procedure, it should
in principle be possible to
extend the results of \autoref
{sec:variance}
and
compute the exact variance of the estimate~$X$
if the whole graph $G$ with all forward and backward proportions is
given.
The general idea is that each condensation step contributes to the
variance, because it replaces a set $S$ by a single sample.

Knuth
\cite[Section 7.2.2.9, 
Theorem~V]{kn4c}
has given a method to compute the variance
for Chen's version 
\cite{Chen-sampling-92} of Bucket Sampling, where the heuristic
function partitions the nodes into strata in an acyclic way.
It would be interesting to get a similar analysis for Crowd Sampling.

\fi

\iflong

\subsection{Relation to dynamic programming}
\label{relation-DP}

Many counting algorithms
are based on \emph{dynamic programming}, or
the \emph{transfer-matrix method}
as it is sometimes called
depending on the implementation.
We show that these algorithms can be seen as a special case of Bucket Sampling.

We illustrate this
with the problem of counting {directed lattice animals},
following Conway~%
\cite
{Conway-2017}.
A \emph{directed animal}
(see \autoref{fig:a-directed-animal})
is a polyomino in the first quadrant such
that every occupied cell, except the cell at the bottom-left
corner, 
has an occupied cell to
its left or at its bottom (or both).
Directed animals are
theoretically well-understood \cite{directed-animals-1998},
and they are chosen here 
for
illustration only.

\begin{figure}[htb]
  \centering
  \includegraphics[scale=1]{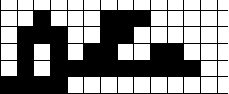}
  \caption{A directed animal with 29 cells}
  \label{fig:a-directed-animal}
\end{figure}

\autoref{fig:directed-animals}a shows an \emph{enumeration tree} $T$
for the directed
animals. 
The animals have been turned by $135^\circ$ and grow from top to bottom.
The black cells are the cells that are already part
of the animal, 
and the shaded cells are the adjacent cells below the chosen cells
whose inclusion has not been decided yet.
At each tree node, we take the
leftmost undecided cell in the highest row (shaded in blue)
and decide whether it is excluded (marked by a cross) or
included (turned black). If there is only a single shaded cell, it has
to be included.
In this way, we try out all possibilities of generating a directed
animal.
We stop when we have 4 black cells, and hence
the 13 leaves of this tree represent the directed animals of size~4.

\begin{figure}[htb]
  \centering
  \includegraphics[scale=0.97]{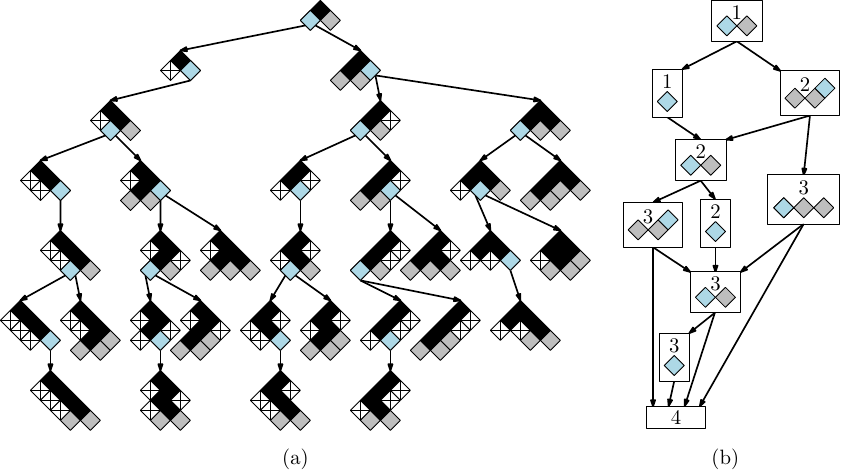}
  \caption{(a) an enumeration tree $T$ for directed animals  and (b) the associated dynamic-programming
    DAG $G$, adapted from Conway \cite[Fig.~1 and~2]{Conway-2017}}
  \label{fig:directed-animals}
\end{figure}

The main idea of dynamic programming is that we don't need to remember
the whole animal 
as it is built:
For determining how to continue,
it suffices to know the \emph{pattern} of
the shaded (grey and blue) cells and the \emph{number} of black cells.
Nodes for which this information agrees have equivalent subtrees.
By merging such nodes,
we arrive at
the
DAG $G$ shown in
\autoref{fig:directed-animals}b.
The directed animals correspond now to the source-
sink paths in
~$G$.


The DAG~$G$ is smaller than the tree $T$ from which it originated.
(For larger examples than the one shown in the picture, the reduction
in size
is more impressive.)
However, 
for the single-path algorithm (Algorithm P), it makes
\emph{absolutely no difference
} whether it moves in the tree $T$
to count leaves according to Theorem~\ref{th:unbiased}
or in the DAG $G$ with update
rule (U$_1$) to count source-sink paths
according to Theorem~\ref{th:unbiased-unweighted-paths}.
(In this spirit,
Conway's reference to the DAG in \cite[Fig.~2]{Conway-2017} as a ``tree''
 is perhaps not
completely off-track.) 


 We can use Bucket Sampling with a heuristic function $h$ that maps the nodes of
 the tree~$T$ to the corresponding nodes in the DAG $G$.
It turns out that,  in this way, we arrive
essentially 
at the
classic
dynamic-programming algorithm for counting objects by the DAG:
We scan the DAG top-down, adding the weight of each node to all its
children.
Since only nodes with equal 
subtrees are thrown together in one bucket,
the algorithm gives the exact 
result.

But the possibilities of Bucket Sampling 
go beyond this:
We are allowed to merge nodes
whose subtrees are different (but hopefully at least similar).
The transition from the tree to the DAG,
which is crucial for the speed-up in dynamic programming,
is, however, not essential for Bucket
Sampling, because it is always implicitly included.
Perhaps,
 advantages of
working with a 
DAG instead of an enumeration tree will
show up 
in connection with other Team Sampling
strategies.

\fi
\section{Extensions and potential applications}
\label{sec:extensions}

\iflong
\subsection{Extension to arc costs}
\label{sec:ext}

For a graph where a cost $c(uv)$ is associated to each arc
$uv$,
the extension to estimating the total {sum} of all arc costs is straightforward:
When proceeding from $u$ to $v$, the algorithm adds
 $W\cdot c(uv)/\pforw_{uv}$
to $X$.
Conceptually, we could simply subdivide every arc and let the newly inserted nodes
take over the role of the arcs.

\fi

\subsection{Some applications}
\label{sec:applications}

We list a few applications where it might be promising to try
some of our estimation algorithms.
 The number of vertices of a polytope, to be discussed in the
first subsection, appears very
natural; most of the remaining applications are more speculative.

\subsubsection
{Vertices of a polytope}
For a polytope, 
a linear
objective function 
induces an orientation of the graph of the
polytope\iflong:
It directs each edge in the direction of improvement of
the objective function, 
\else
, \fi
provided that the
\iflong
objective
\fi
function is
\iflong
sufficiently generic
\else
 generic enough
\fi
and
\iflong
does not assign the same value
\else
 assigns different values
\fi
to adjacent vertices.

As mentioned in \autoref{sec:start},
Avis and Devroye have used the tree version of Algorithm~P to estimate
the number of vertices of a polytope~\cite{Avis-Devroye-2000}.
The simplex algorithm generates a path from each vertex of a polytope
to the optimum vertex, and thereby defines a spanning tree of the
vertices.
This spanning tree is at the core of the
\emph{reverse search} method for enumerating all vertices, and
it is the tree to which
Avis and Devroye apply Algorithm~\hbox{P}.
Salomone, Vaisman, and Kroese~\cite{salomone16} have applied the
Stochastic Enumeration method (which we have called Crowd Sampling) to this tree.

However, the graph of a polytope whose edges are oriented by an
objective function is
naturally a DAG.\footnote
{Every Klee--Minty cube is of course a polytope, but also the
  permutation DAG is obtainable as the graph of a polytope,
  namely the \emph{permutahedron}:
The vertex coordinates of the permutahedron
are given by the inverse permutations of the permutations that form
the nodes of the permutation DAG as described here.
  For any generic objective function, the directed graph of the
  permutahedron edges is isomorphic to the  permutation~DAG.}
 We intend to compare our approach
against the tree-based estimation methods of \cite{Avis-Devroye-2000}
and~\cite{salomone16}.\iflong
\par
{
  When pivoting from vertex to vertex by the simplex algorithm,
  one has to take into account an unpleasant technicality:
The objects that the simplex algorithm visits
are 
the \emph{feasible bases}, and
for polytopes that are not simple, a vertex can correspond to several
bases.
The standard way to treat such \emph{degenerate} or nonsimple polytopes,
where more than $n$ facets meet in a vertex, is to
perturb the constraints, resulting in a simple polytope, and
count only one of the many vertices that correspond to a vertex of the
original polytope. See
\cite{Avis-Devroye-2000} for details.}
\fi
\subsubsection{Triangulations of a planar point set}
\label{sec:triangulations}
Alvarez and Seidel
\cite{alvarez-seidel-13} have given an algorithm with runtime
$O(2^nn^2)$
for counting the
triangulations of an $n$-point set in the plane.
They model triangulations
as
source-sink paths in a DAG with $n2^{n-3}$ nodes.
Our algorithms can therefore be readily applied. 
The quality of the resulting approximation remains to be explored.
The state of the art about approximate counting of triangulations is
rather weak: There is only
a
subexponential-time
approximate counting method
with a subexponential multiplicative error~\cite{alvarez2015-counting}.
\iflong
The current fastest algorithm for exact counting takes time $n^{O(\sqrt
  n)}$~\cite{Till-cactus-2016}.
\fi

\subsubsection
{Linear extensions of a two-dimensional partial order}
\label{sec:linear-extensions}
If we start from an arbitrary permutation $\pi=(\pi_1,\ldots,\pi_n)$
in the permutation DAG, the 
reachable nodes are
the linear
extensions of a two-dimensional partial order $\prec$, which is
defined by the rule
$i\prec j :\Leftrightarrow i<j \land \pi_i<\pi_j$.
Thus our path sampling methods can be used to estimate
the number of linear extensions of such a partial order.
The computation of indegrees must be adjusted to count only those
predecessor permutations that can be reached from $\pi$.
\iflong\else\looseness -1 \fi


Counting the linear extensions is
\#P-complete, even for a two-dimensional partial order, as proved by
Dittmer and Pak
\cite[Theorem~1.3]{dittmer-pak2018}.\footnote{
\iflong  There is a \else A \fi
  subsequent
  journal 
  article by
  the same authors with the same title\iflong, which, however,
omits 
  this result
 and does not even mention it.
 \else\
 does not mention  this result.
 \fi
}
In any case, there already exists a fully polynomial randomized approximation
scheme
for the number of linear extensions of an arbitrary poset,
because the number of linear extensions is directly related to the volume of
the order polytope, and the polytope volume can be approximated by
Markov chain Monte Carlo methods,
see for example
\cite[Section~3]{brightwell91} for details and references.
Our path sampling methods might lead to alternative approximation
schemes that are simpler.


\subsubsection{Pattern-avoiding permutations}
\label{sec:pap}
A permutation $\tau$ \emph{contains} a permutation $\sigma=(\sigma_1,\ldots,\sigma_k)$,
the \emph{pattern}, if $\tau$ has a subsequence of length $k$ whose
elements are in the same relative order as the elements of~$\sigma$;
otherwise, $\tau$ \emph{avoids}~$\sigma$.
For example, the permutation 14253 from \autoref{fig:perm-path}
contains the pattern 1324, in the form of the subsequence 1425,
whereas its predecessor 14523 on the path avoids~1324.
By the Marcus--Tardos theorem, the number of permutations of
length~$n$ avoiding a fixed pattern grows only exponentially in~$n$
\cite{marcus-tardos-04}.
For the patterns of length~3, these numbers are
the Catalan numbers, and for length~4, there are three
essentially different possibilities, so-called
Wilf-equivalence classes, represented by
the patterns
1234, 1342, and 1324
\cite{vatter-15-permutation-classes}.
The 1324-avoiding permutations are the only ones among them
for which no formula is known,
and
they have become a notorious benchmark problem
for counting algorithms.
Their numbers are known exactly up to $n=50$
\cite{conway-guttmann-zinn-justin-18},
and the exponential growth rate is estimated as $11.600\pm 0.003$
by series analysis of these numbers,
while the best proven bounds are $10.27$ and $13.5$
\cite{bevan-brignall-elvey-price-pantone-20}.
\iflong
Markov chain Monte Carlo methods have been used to sample random
1324-avoiding permutations~\cite{madras-liu-10}.
\fi

The class of $\sigma$-avoiding permutations
is closed under taking
subsequences in the following sense:
If we delete an element from a $\sigma$-avoiding permutation and
renumber the remaining elements consecutively, the result
still avoids~$\sigma$.
Hence, the $\sigma$-avoiding permutations of length at most $n$ form
a DAG with the singleton permutation $1$ as the source, in which the
successors of a permutation $\pi$ of length $m<n$ are the
$\sigma$-avoiding permutations of length $m+1$
obtained by inserting a half-integer element at an arbitrary position and renumbering.
Every node can be reached from the source.
Moreover, assuming $k\ge2$,
one can always append either a new largest or smallest element at
the end without creating an occurrence of $\sigma$.
Therefore,
every source-sink path ends at length~$n$.
With costs $c(v)=1$ for the permutations of length~$n$ and $c(v)=0$
otherwise, Algorithm~P estimates their number.
The indegree is easy to compute and requires no test for
pattern containment:
Deleting either of two adjacent elements whose values differ by~1
gives the same result,
and one can show that these are the only coincidences.
Thus, for a permutation $\tau$ of length~$m$, 
\[
  d^-(\tau) = m - \bigl|\{\,i : |\tau_i-\tau_{i+1}|=1\,\}\bigr|.
\]
For the outdegree, one must try the $(m+1)^2$ ways of inserting a
new element into $\tau$, discard duplicates (by the same rule), and
check whether the new element creates an occurrence of~$\sigma$.

If one only allows the insertion of a new \emph{largest} element, one
gets instead a tree, the so-called \emph{generating tree}
\cite{west-95-generating-trees},
to which the original tree version of Algorithm~P is applicable.
It would be interesting to see which of the two settings (tree or DAG) leads to
better estimates.
The exact counting algorithms for 1324-avoiding permutations are, in essence,
dynamic-programming algorithms that merge nodes of the generating tree
whose subtrees are equivalent~\cite{marinov-radoicic-03,Conway-2017}.
As discussed in \autoref{relation-DP} for directed animals,
Bucket Sampling offers a way to relax this:
nodes whose subtrees are only ``similar'' can be put into the
same bucket.
Since the exact counts up to $n=50$ are available for comparison,
the 1324-avoiding permutations are an ideal testing ground
for our methods,
which could then be applied to
get estimates for larger $n$,
or to estimate the counts for other patterns or for sets of several
forbidden patterns, for which no specialized algorithms have been designed.

\subsubsection{Polyominoes}
\label{sec:polyominoes}

Polyominoes are another classic counting problem
of this kind;
the directed animals of \autoref{relation-DP} are
a special type of polyominoes.
As for the 1324-avoiding permutations,
the record exact counts for various types of polyominoes have been
obtained by dynamic programming,
known in this area as the \emph{transfer-matrix} method,
which avoids enumerating the polyominoes individually,
see
\cite{Polyominoes-to-70-ALENEX,Polyominoes-to-70-Algorithmica} for
recent advances, and
\cite{Conway-2017} for an overview that covers both problems.
As discussed in \autoref{relation-DP},
such algorithms can be interpreted as
accumulating (possibly weighted) source-sink paths in a DAG.
They are run on huge DAGs, exhausting the
limits of computer time and memory.
With our approach, the same ideas can be leveraged to get estimates far beyond
the polyomino sizes for which exact computation is feasible.
This might be an alternative to
Markov-chain based sampling techniques~\cite{PhysRevLett.41.1333,Hsu_2005}.
With appropriate
weights $c(v)$, one can also gather statistics about various
properties of polyominoes.

\subsubsection{Maximal sets of an independence system}
\label{sec:maximal-IS}

Honcharov, Sar{\i}y\"uce, Laishram, and
Soundarajan~\cite{honcharov-etal-2023}
have independently discovered the generalization of Algorithm~P
from trees to DAGs
for a special graph problem: They used it to estimate the
number of \emph{skeletal cores} of a graph,
 which
are minimal subgraphs that preserve certain properties.
More precisely, they have the property that
every vertex
keeps its original \emph{core number}, where the core number of
a vertex $v$ is the largest $k$ such that $v$ is contained in a
subgraph of minimum degree~$k$.

They apply Algorithm~P 
to a DAG whose nodes are subgraphs and whose
arcs correspond to edge deletions.
The sets of edges that can be deleted
form an \emph{independence system}: a family of sets that is
closed under taking subsets.
A skeletal core corresponds to
a \emph{maximal} set in this family (a \emph{basis} of the
independence system).

In this setting, the DAG is layered,
and a node at distance~$j$ from the source has indegree~$j$.
This regular structure makes the algorithm particularly easy.
Instead of dividing by the indegree at every step,
Honcharov et al.\  divide the final weight by $j!$,
where $j$ is the number of deleted edges, to account for the $j!$ different
paths that reach the same sink.
This gives exactly the same result as
Algorithm~P with the standard proportions and $c(v)=1$ for the sinks.

\subsection{Graphs with cycles}
\label{sec:cycles}

Marchetti-Spaccamela
\cite{marchettispaccamela-89}
 has proposed an interesting extension of
Algorithm~P to arbitrary digraphs.  From each node that is visited,
this method repeatedly starts a random path in the reverse direction of the
arcs and follows it until it arrives at a node that already
lies on that path or at the source. The number of trials necessary
to reach the latter outcome
is used to adjust the estimate.  Since the probability of this event
can be very low, the method requires many samples, and it is practical
only in limited settings.

\section{Conclusion}
\label{sec:conclusion}

We have proposed Team Sampling, a general scheme for obtaining
unbiased estimates of the size of a DAG with a team of cooperating
agents, and we have tried two specializations, Bucket Sampling and
Crowd Sampling, which carry over methods of
Chen~\cite{Chen-sampling-92} and Rubinstein~\cite{Rubinstein-2013}
from trees to DAGs.
There is
ample opportunity to come up with more ideas and try them out; we have
barely scratched the surface with our experiments.
For example, one can
combine the ideas of Bucket Sampling and Crowd Sampling, allowing
up to $B$ nodes to occupy one bucket.

\subsection{Choice of forward and backward proportions}
\label{sec:choice}
For the choice of forward proportions $\pforw_{uv}$,
Importance Sampling
gives some guidelines about what one wants to achieve
(see \autoref{sec:IS}).

If the 
backward proportions $\qback_{uv}$ are nonnegative, they
could of course be interpreted as transition probabilities for
a random walk in the reverse graph.
Nevertheless, their role in the algorithm
is not
understood at all.
\iflong
One may even allow them to become negative or zero:
$\qback_{uv}=0$ effectively means that the sample is ignored when the
algorithm traverses the arc $uv$.
Negative proportions
$\qback_{uv}<0$ are 
mysterious.

\fi

\subsection{Provable approximation quality?}
Vaisman and Kroese
showed that
 Stochastic Enumeration (i.e., Crowd Sampling)
is a fully polynomial randomized approximation scheme (FPRAS)
for counting the number of nodes of trees,
for \emph{most} trees in a certain (rather restricted) class of random
trees
\cite[Theorem 4.3\iflong\footnote
{In the statement of this theorem, the budget parameter $B$ for the algorithm is not
  specified. It must be set to $kB'$, where $B'$ is
 defined in
  \cite[Theorem 4.1, p.~52]{Vaisman-Kroese-2017-trees},
  depending polynomially on the tree height and the maximum degree~$k$, as well
  as the
  parameters of the random tree model (confirmed by Radislav Vaisman, personal
  communication, July 2025).}\fi
]{Vaisman-Kroese-2017-trees}.
It would be interesting to see whether Path Sampling or Team Sampling
 leads to an FPRAS for some combinatorial counting problems.
Some candidate problems were listed in
\autoref{sec:applications}.


\iflong
\paragraph{Acknowledgment.}

The ideas of this paper were inspired by discussions initiated at
the workshop on \emph{Computational Polyhedral
Geometry},
which took place
on June 26, 2025
as part of the
Computational Geometry Week (CG Week) in Kanazawa, Japan.\footnote{\url{https://socg25.github.io/workshops.html}}
I thank  Laura Casabella and Dante Luber
for organizing 
this workshop, and the other
participants for the lively discussions.

I thank Bruno Scherrer for making me aware that the
idea of Algorithm~P, which I developed independently,
had been anticipated a long time ago by Leonard Pitt \cite{Pitt1987}.

\fi

\paragraph{AI usage statement.}
The AI assistant Claude from Anthropic has been used
to check
the manuscript for consistency
and correctness, and to suggest improved formulations in many places.
The AI assistant pointed the author to 
Kitahara and Mizuno~\cite {KITAHARA-Klee-Minty-2011},
where an earlier version of a Klee--Minty cube with integral
vertices was used.
 Section~\ref{sec:pap}
 has been generated by the AI assistant from a stub after
 the author had forgotten what he intended to write there.
The programs were initially written and the test runs were carried out, alas,
without any AI assistance.
The AI assistant found a programming error, by which
\autoref{tab:KM-IS1} in \autoref{sec:importance} reported results of the same experiment as
\autoref{tab:KM-IS},
 and
two inconsistencies between the programs and the descriptions in the text;
all these mistakes have been corrected.
The author assumes the 
responsibility for the results.

\bibliographystyle{plainurl}
\bibliography{estimating-DAG}

\appendix

\iflong\else

\fi

\section{Importance Sampling for Klee--Minty cubes}
\label{sec:importance}

As mentioned in \autoref{sec:IS}, we have tried Importance Sampling by
outdegree also for Klee--Minty cubes.
In comparison to
\autoref{tab:P-KM},
the results of
\autoref{tab:KM-IS} show a substantial improvement in the
theoretical standard deviation, by several orders of magnitude.
Like for permutations,
the comparison in terms of
empirical behavior is not at all clear. Even with 10 million samples, the
fluctuations are too high to make a reasonable assessment.

\begin{table}[htb]
  \centering
  \par\noindent\begin{tabular}{|c|r@{}lll|r@{}lll|}
\hline
\multicolumn{9}{|c|}{Klee--Minty cubes}\\
\hline
&\multicolumn4{c|}{$n=10$}&\multicolumn4{c|}{$n=12$}\\
$m$ & mean &\,$\pm$\,std.\,dev.&\multicolumn 1 c{max.}&\multicolumn 1 {c|}{$\hat\chi$}& mean &\,$\pm$\,std.\,dev.&\multicolumn 1 c{max.}&\multicolumn 1 {c|}{$\hat\chi$}\\
\hline\vbox to 11pt{}%
$10^1$&$65.0$&$\,\pm\,46.7$& $1.45{\times}10^{2}$ & $0.22366$&$503.5$&$\,\pm\,668.7$& $1.99{\times}10^{3}$ & $0.39488$  \\
$10^2$&$500.7$&$\,\pm\,1{,}279.9$& $7.67{\times}10^{3}$ & $0.15319$&$685.5$&$\,\pm\,1{,}676.9$& $1.05{\times}10^{4}$ & $0.15256$  \\
$10^3$&$649.7$&$\,\pm\,3{,}306.4$& $5.54{\times}10^{4}$ & $0.08534$&$1{,}854.5$&$\,\pm\,15{,}942.5$& $3.76{\times}10^{5}$ & $0.20282$  \\
$10^4$&$880.2$&$\,\pm\,11{,}958.6$& $8.46{\times}10^{5}$ & $0.09610$&$2{,}418.5$&$\,\pm\,33{,}458.9$& $2.61{\times}10^{6}$ & $0.10777$  \\
$10^5$&$910.0$&$\,\pm\,37{,}608.7$& $1.16{\times}10^{7}$ & $0.12699$&$6{,}892.9$&$\,\pm\,1.23{\times}10^{6}$& $3.86{\times}10^{8}$ & $0.55991$  \\
$10^6$&$983.9$&$\,\pm\,87{,}220.8$& $8.33{\times}10^{7}$ & $0.08466$&$3{,}969.2$&$\,\pm\,5.15{\times}10^{5}$& $3.86{\times}10^{8}$ & $0.09723$  \\
$10^7$&$1{,}006.4$&$\,\pm\,1.45{\times}10^{5}$& $4.38{\times}10^{8}$ & $0.04347$&$3{,}737.7$&$\,\pm\,4.70{\times}10^{5}$& $7.03{\times}10^{8}$ & $0.01881$  \\
\hline
\vbox to 11pt{}&$1{,}024\phantom{.0}$&$\,\pm\,\rlap{$2.141{\times}10^{6}$}$&&&$4{,}096\phantom{.0}$&$\,\pm\,\rlap{$7.762{\times}10^{10}$}$&&\\
&\multicolumn4{l|}{\#visited: avg\ = 22.4, max\ = 86}&\multicolumn4{l|}{\#visited: avg\ = 29.2, max\ = 104}\\
\hline
\end{tabular}

  \caption{Algorithm P for Klee--Minty cubes with Importance Sampling
    by outdegree}
  \label{tab:KM-IS}
\end{table}

\subsection{Importance Sampling by the size of the reachable set}
\label{sec:reachable-set}

As mentioned,
when Path Sampling is applied to a rooted tree,
the ideal forward
proportions 
are known:
The probabilities $\pforw_{uv}$ should be proportional to the cost of
the subtrees rooted at the children $v$ of node~$u$.
For DAGs, it is not clear whether such a perfect choice of
forward proportions that completely eliminates variations of $X$ exists. As a substitute for the size of
the subtree rooted at $v$, we may take the size $R_v$ of the node set
reachable from $v$.

For Klee--Minty cubes, we are in the lucky situation that
the size $R_v$ of the reachable set is readily available:
The Klee--Minty cube of dimension $n$ has a directed Hamilton path, in which
the last coordinate $x_n$ increases in an arithmetic progression
$0,1,2,\ldots,2^n-1$, see
\autoref{fig:klee-minty}. Thus, $R_v=2^n-x_n(v)$.

 We made two experiments where
 $R_v$ is taken into account for computing the forward proportions.
We emphasize 
that this
is not suitable in practical applications:
If the size of the reachable set were known, 
there would be no
need to compute an estimate by running a sampling algorithm. 
However, the experiment teaches us what we might be looking for in
the choice of forward proportions.

In the first experiment (\autoref{tab:KM-IS1}), we simply
chose the probabilities $\pforw_{uv}$
for the successors $v$ of a node~$u$ proportional to~$R_v$.
We see no improvement over simply using the outdegree.
For $n=12$, both the observed and the theoretical standard deviation
are even about one order of magnitude larger.

In the second experiment (\autoref{tab:KM-IS2}), we took into
account the fact that a node $v$ with large indegree is visited from
several predecessors. We therefore chose the probabilities
proportional to
$R_v/d^-(v)$.
Here we get
much better estimates, and we see a
substantial reduction of the standard deviation. While these weights
are not perfect, they indicate the direction in which one should
look.
Conceivably, a perfect choice of parameters can be obtained by
adjusting also the backward proportions $\qback_{uv}$.

\begin{table}[htb]
  \centering
  \par\noindent\begin{tabular}{|c|r@{}lll|r@{}lll|}
\hline
\multicolumn{9}{|c|}{Klee--Minty cubes}\\
\hline
&\multicolumn4{c|}{$n=10$}&\multicolumn4{c|}{$n=12$}\\
$m$ & mean &\,$\pm$\,std.\,dev.&\multicolumn 1 c{max.}&\multicolumn 1 {c|}{$\hat\chi$}& mean &\,$\pm$\,std.\,dev.&\multicolumn 1 c{max.}&\multicolumn 1 {c|}{$\hat\chi$}\\
\hline\vbox to 11pt{}%
$10^1$&$218.6$&$\,\pm\,226.5$& $7.47{\times}10^{2}$ & $0.34168$&$472.9$&$\,\pm\,661.9$& $2.03{\times}10^{3}$ & $0.43003$  \\
$10^2$&$640.6$&$\,\pm\,1{,}830.8$& $1.30{\times}10^{4}$ & $0.20285$&$2{,}626.8$&$\,\pm\,11{,}882.5$& $8.69{\times}10^{4}$ & $0.33074$  \\
$10^3$&$756.1$&$\,\pm\,5{,}094.5$& $1.33{\times}10^{5}$ & $0.17623$&$2{,}167.0$&$\,\pm\,27{,}495.6$& $8.44{\times}10^{5}$ & $0.38950$  \\
$10^4$&$818.4$&$\,\pm\,9{,}017.6$& $5.12{\times}10^{5}$ & $0.06258$&$2{,}403.3$&$\,\pm\,36{,}794.5$& $2.26{\times}10^{6}$ & $0.09389$  \\
$10^5$&$891.8$&$\,\pm\,20{,}241.5$& $5.13{\times}10^{6}$ & $0.05748$&$4{,}349.1$&$\,\pm\,5.99{\times}10^{5}$& $1.88{\times}10^{8}$ & $0.43315$  \\
$10^6$&$984.7$&$\,\pm\,31{,}787.0$& $1.26{\times}10^{7}$ & $0.01279$&$6{,}346.5$&$\,\pm\,2.53{\times}10^{6}$& $2.35{\times}10^{9}$ & $0.37095$  \\
$10^7$&$982.6$&$\,\pm\,50{,}113.4$& $6.09{\times}10^{7}$ & $0.00620$&$4{,}877.9$&$\,\pm\,3.17{\times}10^{6}$& $9.30{\times}10^{9}$ & $0.19071$  \\
\hline
\vbox to 11pt{}&$1{,}024\phantom{.0}$&$\,\pm\,\rlap{$2.080{\times}10^{6}$}$&&&$4{,}096\phantom{.0}$&$\,\pm\,\rlap{$6.025{\times}10^{11}$}$&&\\
&\multicolumn4{l|}{\#visited: avg\ = 37.8, max\ = 108}&\multicolumn4{l|}{\#visited: avg\ = 53.0, max\ = 148}\\
\hline
\end{tabular}

  \caption{Algorithm P for Klee--Minty cubes with Importance Sampling
    by the size of the reachable set}
  \label{tab:KM-IS1}
\end{table}

\begin{table}[htb]
  \centering
  \par\noindent\begin{tabular}{|c|r@{}lll|r@{}lll|}
\hline
\multicolumn{9}{|c|}{Klee--Minty cubes}\\
\hline
&\multicolumn4{c|}{$n=10$}&\multicolumn4{c|}{$n=12$}\\
$m$ & mean &\,$\pm$\,std.\,dev.&\multicolumn 1 c{max.}&\multicolumn 1 {c|}{$\hat\chi$}& mean &\,$\pm$\,std.\,dev.&\multicolumn 1 c{max.}&\multicolumn 1 {c|}{$\hat\chi$}\\
\hline\vbox to 11pt{}%
$10^1$&$607.9$&$\,\pm\,396.9$& $1.28{\times}10^{3}$ & $0.21010$&$868.0$&$\,\pm\,1{,}456.5$& $4.92{\times}10^{3}$ & $0.56699$  \\
$10^2$&$857.0$&$\,\pm\,1{,}552.7$& $1.27{\times}10^{4}$ & $0.14858$&$1{,}704.5$&$\,\pm\,4{,}347.4$& $4.04{\times}10^{4}$ & $0.23701$  \\
$10^3$&$985.6$&$\,\pm\,3{,}287.0$& $5.65{\times}10^{4}$ & $0.05729$&$3{,}150.8$&$\,\pm\,13{,}846.9$& $2.47{\times}10^{5}$ & $0.07848$  \\
$10^4$&$971.5$&$\,\pm\,3{,}667.1$& $1.16{\times}10^{5}$ & $0.01198$&$3{,}677.1$&$\,\pm\,37{,}719.8$& $3.33{\times}10^{6}$ & $0.09052$  \\
$10^5$&$1{,}031.9$&$\,\pm\,8{,}148.2$& $1.93{\times}10^{6}$ & $0.01874$&$3{,}717.0$&$\,\pm\,52{,}576.6$& $1.26{\times}10^{7}$ & $0.03399$  \\
$10^6$&$1{,}029.4$&$\,\pm\,11{,}760.9$& $9.25{\times}10^{6}$ & $0.00899$&$4{,}268.1$&$\,\pm\,2.40{\times}10^{5}$& $2.18{\times}10^{8}$ & $0.05108$  \\
$10^7$&$1{,}025.8$&$\,\pm\,9{,}644.3$& $1.21{\times}10^{7}$ & $0.00117$&$4{,}053.0$&$\,\pm\,1.22{\times}10^{5}$& $2.18{\times}10^{8}$ & $0.00538$  \\
\hline
\vbox to 11pt{}&$1{,}024\phantom{.0}$&$\,\pm\,\rlap{$1.071{\times}10^{4}$}$&&&$4{,}096\phantom{.0}$&$\,\pm\,\rlap{$9.419{\times}10^{5}$}$&&\\
&\multicolumn4{l|}{\#visited: avg\ = 42.5, max\ = 112}&\multicolumn4{l|}{\#visited: avg\ = 58.7, max\ = 152}\\
\hline
\end{tabular}

  \caption{Algorithm P for Klee--Minty cubes with Importance Sampling
    by the size of the reachable set divided by the indegree}
  \label{tab:KM-IS2}
\end{table}

\iflong
\section{Details about implementing Klee--Minty cubes}
\label{sec:KM-details}

The natural way to associate a bit vector $(y_1,\ldots,y_n)$ to the
vertices of a Klee--Minty cube
\eqref{eq:KM} (\autoref{sec:klee-minty})
is derived from the standard cube:
$y_i=1$ if $x_i$ is at the ceiling, i.\,e.,
$x_i=2^i-1-x_{i-1}$,
and $y_i=0$ if $x_i$ is at the floor, i.\,e.,
$x_i=x_{i-1}$.
In this way, the Klee--Minty cube is mapped to the standard cube, and
the $n$ neighbors of
$(y_1,\ldots,y_n)$ are obtained by simply flipping a single bit.
The downside of this representation is that it is not
easy to see whether an edge is improving.

Our program uses
a different representation $(z_1,\ldots,z_n)$, which was introduced by
Gärtner, Henk, and Ziegler
\cite{Randomized-simplex-Klee-Minty-1998}:
%
%
\begin{align*}
  z_n &= y_n \\
  z_{n-1} &= y_n\oplus y_{n-1} \\
  z_{n-2} &= y_n\oplus y_{n-1}\oplus y_{n-2} \\
  &\ \vdots\\
  z_{1} &= y_n\oplus y_{n-1}\oplus y_{n-2} \oplus
                  \cdots \oplus y_2 \oplus y_1
\end{align*}
It is the sequence of partial sums of the $y_i$'s using addition
modulo~2, or the
exclusive-or operation~$\oplus$.

The outgoing arcs are identified by
the zero bits in this representation.
If we follow such an arc, we must flip the $i$-th bit together with
all lower-order bits.
The reason is that while
the orientation of all incident edges with higher index is maintained,
 the lower-index directions are perverted.

More concretely, we represent a vertex by the $n$-bit 
integer
$(z_nz_{n-1}\ldots z_2z_1)_2$. It turns out that the value of this
integer is equal to the $x_n$-coordinate of the vertex.

\fi

\iflong
\section{Purdom's partial backtracking method}
\label{sec:purdom}
 Purdom
\cite{purdom-78}, in
1978, suggested
the following
variation of Knuth's Path Sampling
 algorithm for trees: Instead of proceeding to a single child, proceed
 to at most $m$ randomly selected 
 children.
 The parameter $m$ can itself be chosen at random.
 In particular, Purdom looked at the setting of binary trees,
 and he proposed to choose $m=2$ or $m=1$ with probability $p$ and
 $1-p$ respectively.
\looseness-1

 Purdom 
 analyzed the number of visited nodes as well as the variance of
 the estimated tree size,
 for a certain class of random binary trees.
He also made experiments with
a backtracking program.
\looseness-1

This method could be easily adapted to DAGs, but it will probably not
compete with the more sophisticated Team Sampling methods, which allow
a more direct control of the number of visited nodes.
%
Still, one sentence from the conclusion of the paper is worth citing:
``The analysis suggests that good results can be obtained by choosing
the number of branches to investigate so that each level, from the
root down to the bulk of the nodes, is examined at about the same
frequency.''
This is in line with the strategy of the Crowd Sampling scheme,
which sets a fixed budget $B$. In case of a tree,  Crowd Sampling 
would
indeed look at up to $B$ nodes in each level.

\fi

\iftrue


\section{Index of important notation}
\label{sec:notation}
\begingroup \parindent 0pt
\noindent\null\indent
$A,A',\hat A$ active set (list, or multiset) of nodes in Team Sampling
(\autoref{sec:TS}) and its variations


$B$ budget: bound on the size of the active set in Crowd Sampling
(\autoref{sec:crowd})

$c(v),c(uv), c(G)$ node costs,  arc costs, total cost of the graph
(\autoref{sec:start}, \autoref{sec:ext})

$c_e(G)$ total cost with path weights, the expectation of $X$
(\autoref{sec:arbitrary-arc}, \autoref{th:unbiased-paths})

$\hat \chi$  Chatterjee--Diaconis score
(\autoref{sec:permutations})

$d^+(v),d^-(v)$ outdegree, indegree
(\autoref{sec:start})

$D_u(W)$ random 
value of all future additions to $X$
when $u$ is entered with weight $W$
(\autoref{sec:variance})

$e_{uv}, e(P)$ arc weights, path weights
(\autoref{sec:arbitrary-arc});
assumed positive in \autoref{sec:team}


$E[\ ]$ expectation

$F_u$ random factor with $D_u(W)=F_uW$
(\autoref{sec:variance})

$f_v =\sum_{w} \mathcal P(v\to w)c(w)=E[F_v]$, coefficients in the
invariant for Team Sampling
(\autoref{sec:TS-correctness})

$G=(V,E)$ the input DAG with node set $V$ and arc set $E$
(\autoref{sec:start})

$h\colon V\to U$  heuristic function, which assigns nodes to buckets
(\autoref{sec:bucket})



$m$ number of repetitions in the experiments
(Section~\ref{sec:permutations})

$\mu_j$ mean of $F_{v_j}$ for a successor node $v_j$
(\autoref{sec:variance})

$n$ 
dimension of the Klee--Minty cube
(\autoref{sec:klee-minty});
number of permuted elements 
(Sections~\ref{sec:permutations},~\ref{sec:pap})



$\mathcal P(v\to w)$ sum of the path weights $e(P)$ over all paths $P$
from $v$ to $w$
(\autoref{sec:TS-correctness})

$P=u_0u_1\ldots u_k$ path
(\autoref{sec:arbitrary-arc})




$\Pr[\ ]$ probability

$\pforw_{uv}$
forward proportions
(Section~\ref{sec:weighted});
standard choice $\pforw_{uv}=1/d^+(u)$

$\qback_{uv}$
backward proportions
(Section~\ref{sec:weighted});
standard choice $\qback_{uv}=1/d^-(v)$



$R_v$ number of nodes reachable from $v$ in a Klee--Minty cube
(\autoref{sec:reachable-set})

$s_0$ source node
(\autoref{sec:start})

$S=\{u_1,\ldots,u_k\}$ subset of $A$ to be condensed in Team Sampling
(\autoref{sec:TS}, \autoref{sec:mix})

$S_1,\ldots,S_i,\ldots,S_B$ the sets that are condensed in Crowd Sampling
(\autoref{sec:crowd})



(U),
(U$_q$),
(U$_1$),
(U$_e$) variations of the weight update step
(\autoref{sec:start})

$U$ universe that indexes the buckets
(\autoref{sec:bucket})



$\var(\,)$ variance
(\autoref{sec:variance})

$\var_j$ variance 
of $F_{v_j}$ for a successor node $v_j$
(\autoref{sec:variance})

$W, W_v$ current weight
(\autoref{sec:start});
the value of $W$ when $v$ is visited
(\autoref{sec:analysis})

$W[v]$ weight of an active node $v\in A$ in the Team Sampling algorithm
(\autoref{sec:TS})


$W_i$, $W'_i$ weights before and after condensation
(\autoref{sec:mix})


$X$
the current accumulated estimate
(\autoref{sec:start});
$X^{\mathrm{current}}$, $X^{\mathrm{final}}$ specific values of $X$
(\autoref{sec:TS-correctness})





\#visited, \#generated, \#saved: statistics in the tables
(Sections \ref{sec:bucket}--\ref{sec:crowd})


\endgroup

\fi

\end{document}